\documentclass[journal, doublecolumn]{IEEEtran}
\usepackage{amssymb}
\usepackage{amsthm}
\usepackage{extarrows}
\usepackage{graphicx}
\usepackage{algorithm}
\usepackage{algorithmic}
\usepackage{cite}
\usepackage{bm}
\usepackage{float}
\usepackage{multirow}
\usepackage{color}
\usepackage{subfigure}
\usepackage{caption}
\usepackage{epstopdf}
\usepackage{hyperref}
\usepackage{fixltx2e}
\usepackage{longtable}
\usepackage{diagbox}
\usepackage{changepage}
\usepackage{amsmath, amsfonts}
\usepackage{dsfont}
\usepackage{makecell}
\usepackage{optidef}

\theoremstyle{definition}

\theoremstyle{plain}
\newtheorem{theorem}{Theorem}
\newtheorem{lemma}{Lemma}
\newtheorem{proposition}{Proposition}
\newtheorem{corollary}{Corollary}

\begin{document}
\title{Single-Waveguide Multiple–Pinching-Antenna Systems: OMA versus NOMA}

\author{Yanyu~Cheng,
Hao~Li,
Zhen~Wang,
and~Zhiguo~Ding,~\IEEEmembership{Fellow,~IEEE}

\thanks{
Yanyu~Cheng and Zhiguo~Ding are with the School of Electrical and Electronic Engineering, Nanyang Technological University, Singapore 639798 (e-mail: ycheng022@e.ntu.edu.sg; zhiguo.ding@ntu.edu.sg).

Hao~Li and Zhen~Wang are with the School of Cyberspace, Hangzhou Dianzi University, Hangzhou 310018, China (e-mail: \{242270075; wangzhen\}@hdu.edu.cn).
}
}
\maketitle

\begin{abstract}
This paper investigates the performance of a pinching-antenna (PA) system with a signal waveguide and multiple pinching antennas to serve users distributed across multiple rooms. The performance of the system is evaluated through a comparative analysis under both orthogonal multiple access (OMA) and non-orthogonal multiple access (NOMA) schemes. Specifically, this paper derives closed-form expressions for the outage probability (OP) and ergodic rate (ER) in each scheme. Furthermore, asymptotic analyses are conducted to characterize the system behavior in the high signal-to-noise ratio (SNR) regime. Extensive Monte Carlo simulations are utilized to validate the accuracy of the analytical derivations. The comparative results can be summarized as follows: 
1) in the downlink fixed-rate scenario, whether OMA or NOMA achieves better outage performance depends on system parameters, such as the number of users and power allocation coefficients; 
2) in the uplink fixed-rate scenario, the outage performance of NOMA is inferior to that of OMA in the high-SNR regime, and the decay rate of the OP for NOMA users depends on the rate thresholds; 
and 3) for both uplink and downlink adaptive-rate scenarios, the rate performance comparison of the two schemes depends on system parameters in the low-SNR regime, whereas OMA generally outperforms NOMA in the high-SNR regime.
\end{abstract}

\begin{IEEEkeywords}
Non-orthogonal multiple access, orthogonal multiple access, pinching-antenna systems.
\end{IEEEkeywords}

\IEEEpeerreviewmaketitle

\section{Introduction}
The widespread deployment of fifth-generation (5G) networks has revolutionized connectivity through innovations such as massive multiple-input
multiple-output (MIMO) and millimeter wave (mmWave) \cite{5G:01, 5G:02}. However, the inherent limitations regarding spectrum scarcity and energy efficiency, combined with the unprecedented demands of emerging applications like holographic telepresence and the internet of everything, are driving the evolution toward sixth-generation (6G) networks \cite{6G:01, 6G:02, 6G:03}.

To address this demand, flexible antenna systems have recently attracted extensive research interest, with representative examples including reconfigurable intelligent surfaces (RIS), movable antenna (MA), and fluid antenna (FA) \cite{RIS:01, RIS:02, MAS:01, FAS:01}.
RIS are planar metasurfaces composed of a large number of programmable reflecting elements, each of which can independently adjust the phase of the incident signal, thereby enabling passive beamforming to control the wireless propagation environment and significantly enhance network coverage and communication performance \cite{IRS:02, IRS:03, extract:01-RIS}.
Along a different line of research, to overcome the inherent limitations of conventional fixed-position antennas (FPAs) in fully exploiting spatial diversity, MA systems have been proposed. By allowing dynamic position adjustment of antenna elements within a confined region to align with favorable channel conditions, the MA system can significantly improve spectral efficiency and interference mitigation \cite{MAS:01, MAS:02, MAS:03, MAS:04}.
Moreover, benefiting from advances in flexible metallic materials, the FA system has been proposed as a reconfigurable technology that employs conductive fluids to dynamically adjust the antenna position and shape \cite{FAS:01}. By exploiting continuous spatial diversity, the FA system effectively mitigates small-scale fading \cite{FAS:02}.

Despite existing research on flexible antenna systems has demonstrated significant potential along various directions, inherent limitations remain in combating free-space path loss and line-of-sight (LoS) blockage \cite{PASS:02}. To address these challenges, NTT DOCOMO proposed the pinching-antenna (PA) system in 2021 as a new flexible communication paradigm \cite{PASS:01}. The PA system employs a dielectric waveguide on which detachable pinching antennas (PAs) are deployed. In this architecture, signals are radiated through controlled leakage from the  PAs without requiring additional radio-frequency resources. Overall, owing to its high deployment flexibility and low implementation cost, the PA system has emerged as a promising solution for enhancing wireless coverage and reliability in complex propagation environments.

\subsection{Related Work}
\subsubsection{Conventional Flexible Antenna Systems}
In \cite{FlexibleAntenna:01}, the authors summarized key advancements in flexible antenna systems, highlighting low cost, high performance, and ease of installation, thereby enabling applications in wearable devices.
In \cite{FAS:01}, the authors analyzed an FA system where a single antenna dynamically switches among multiple candidate positions within a small linear space to achieve the strongest signal, demonstrating that a single-antenna FA system can outperform traditional multi-antenna maximum ratio combining (MRC) systems even within minimal physical space.
Subsequently, in \cite{FAS:03}, the performance limits of FA systems were investigated, showing that a single-antenna FA system can achieve the capacity performance of multi-antenna MRC systems within a confined space.
In \cite{FAS:04}, the authors proposed a machine learning-based framework for channel state information (CSI) estimation and port optimization, revealing the potential of artificial intelligence (AI)-enabled FA systems in wireless communications.
Moreover, the performance gains of MA systems compared to conventional fixed-position antennas were explored in \cite{MAS:01, MAS:04}, specifically regarding signal enhancement and spatial multiplexing.
In \cite{MAS:02, MAS:03}, the research on MA systems has been extended to the field of sensing, thereby underscores the promise of MA technology for wireless systems.
These innovations collectively represent a significant step toward more efficient and versatile wireless networks.

\subsubsection{Pinching Antenna Systems}
In \cite{PASS:01}, an innovative concept was introduced, enabling flexible antenna deployment over large spatial regions, thereby inspiring extensive research in the field of PA systems.
In \cite{PAS:08}, the authors developed analytical frameworks for single- and multi-antenna PA system configurations, deriving closed-form ergodic rate (ER) expressions and proving that multiple PAs enhance data rates. 
Their results highlight the ability of PA systems to establish strong LoS links, mitigate large-scale path loss, and outperform conventional antenna systems in both orthogonal multiple access (OMA) and non-orthogonal multiple access (NOMA) scenarios. 
In \cite{PAS:07}, the authors investigated the impact of PA systems on the user outage probability (OP) across both LoS and non-line-of-sight (NLoS) scenarios. 
The results indicated that while PA systems based transmission exhibits superior stability compared to conventional fixed-antenna systems in LoS links, no significant performance gain was observed in NLoS environments.
Moreover, several studies have shown that deploying multiple PAs within a single space and optimizing their positions and quantities can considerably enhance overall system performance~\cite{PAS:06, PAS:05, PAS:04}. 
An uplink rate-splitting multiple access (RSMA)-based PA system was explored in \cite{PAS:09}, and the authors characterized its outage performance.

\subsection{Motivation and Contributions}
When users are located in different rooms, traditional single-antenna systems struggle to provide simultaneous LoS links to all users. To address this, the deployment of a PA system offers a practical solution. In this context, comparing the performance of OMA, such as time-division multiple access (TDMA), and NOMA remains of significant value. 
Specifically, in TDMA, users transmit at full power but must share time-domain resources. Conversely, in NOMA, users transmit throughout the entire time duration but are required to share the total transmit power.
Therefore, this study addresses a fundamental question: which scheme, OMA or NOMA, provides superior performance under different scenarios? The main contributions of this work are summarized as follows:
\begin{itemize}
    \item We propose a novel PA-based communication model, where multiple PAs are strategically positioned along a single waveguide to establish LoS links for all distributed indoor users in different rooms.
    \item In the fixed-rate system, this paper derives closed-form expressions for the OPs in the downlink for both OMA and NOMA schemes. For the uplink, explicit expressions for OPs are derived. Given the complexity of obtaining closed-form expressions for a general multi-user scenario, this study focuses on a two-user special case. It derives approximate closed-form expressions for OPs. Finally, Monte Carlo simulations are conducted to verify the accuracy and consistency of the theoretical derivations.
    \item In the adaptive-rate systems, this paper derives closed-form expressions for ERs of users in the downlink under both OMA and NOMA schemes. Similar to the analysis of OP, explicit expressions for ERs of users in the uplink are derived, and approximate closed-form expressions for a two-user scenario are derived. Furthermore, the paper develops asymptotic expressions for ERs in the high signal-to-noise ratio (SNR) regime.
    \item This paper compares the performance of OMA and NOMA in both uplink and downlink scenarios. The key insights are summarized as follows: 1) in the downlink fixed-rate scenario, whether OMA or NOMA achieves better outage performance depends on system parameters, such as the number of users and power allocation coefficients; 2) in the uplink fixed-rate scenario, the outage performance of NOMA is inferior to that of OMA in the high-SNR regime, and the decay rate of the OP for NOMA users depends on the rate thresholds; and 3) for both uplink and downlink adaptive-rate scenarios, the rate performance comparison of the two schemes depends on system parameters in the low-SNR regime, whereas OMA generally outperforms NOMA in the high-SNR regime.
\end{itemize}

\section{System Model}

\begin{figure}[t]
\centering
\includegraphics[width=\linewidth]{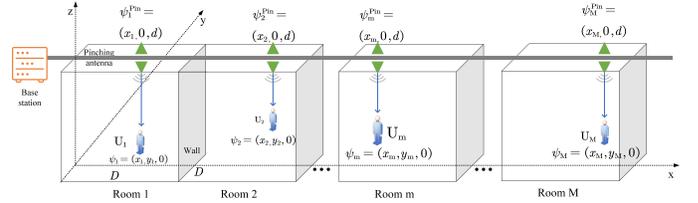}
\caption{A pinching-antenna system with single waveguide and multiple pinching antennas.}
\label{fig_system_model}
\end{figure}
Consider a wireless communication system where a base station (BS) serves $M$ users located in $M$ square rooms of side length $D$, with each room containing exactly one user.
The users, denoted as $\mathrm{U}_1, \mathrm{U}_2, \ldots, \mathrm{U}_M$, are uniformly distributed within their respective rooms. 
To geometrically model the system, a three-dimensional Cartesian coordinate system is established (Fig.~\ref{fig_system_model}), where the location of $\mathrm{U}_m$ is defined as $\psi_m = [x_m, y_m, 0]$.
A waveguide is installed along the ceiling of each room, aligned parallel to the $x$-axis at a height $d$ above the floor ($z = d$). 
To enable simultaneous multiuser transmission, $M$ PAs are deployed on the waveguide, i.e., one PA is deployed in each room.
The key feature of the PA is deployment flexibility. 
For the considered system, the $m$-th PA is placed at the position closest to $\mathrm{U}_m$, i.e., $\psi_m^{Pin} = [x_m, 0, d]$.

\subsection{Downlink}
\subsubsection{OMA}
For the considered system, TDMA is used for the OMA scenario. 
During the $m$-th time slot, $\mathrm{U}_m$ is served, and the data rate of $\mathrm{U}_m$ is given by
\begin{equation}\label{Eq:model-downlink-OMA-rate}
\begin{split}
R_m^{D,OMA} = \frac{1}{M}\log_2\left(1+\frac{\eta \rho_{D}}{|\mathbf{\psi}_m^{Pin}-\mathbf{\psi}_m|^2}\right),
\end{split}
\end{equation}
where $\eta = \frac{c^2}{16\pi^2f_c^2}$, $c$ is the speed of light, $f_c$ is the carrier frequency, $|\mathbf{\psi}_m^{Pin}-\mathbf{\psi}_m|$ is the distance between the PA and $\mathrm{U}_m$, $\rho_D = \frac{P_t}{\sigma^2}$ is the transmit SNR of the BS,  $P_t$ is the transmit power of the BS, and $\sigma^2$ is the additive white Gaussian noise (AWGN) at $\mathrm{U}_m$.

\subsubsection{NOMA}

For simultaneously serving multiple users, the users' signals must to be superimposed as $s=\sum_{m=1}^{M}\sqrt{\alpha_m}s_m$, where $s_m$ is $\mathrm{U}_m$'s signal, $\alpha_m$ is the power allocation coefficient for $\mathrm{U}_m$, with $\alpha_1 > \alpha_2>\dots>\alpha_M$ and $\sum_{m=1}^{M}\alpha_m=1$~\cite{extract:03-NOMA}.
The users are ordered according to their channel conditions in ascending order. 
Without loss of generality, we assume that $\mathrm{U}_m$ is the $m$th user, and $|h_1|^2\leq\cdots\leq|h_M|^2$, where $|h_m|^2 = \frac{\eta}{|\mathbf{\psi}_m^{Pin}-\mathbf{\psi}_m|^2}$.

According to the principle of power-domain NOMA, $\mathrm{U}_m$ decodes $\mathrm{U}_k$'s signal, $1\leq k\leq m$, before decoding its own signal. Consequently, the data rate for $\mathrm{U}_m$ to decode $\mathrm{U}_k$'s signal is given by
\begin{align}\label{Eq:systemmodel-downlink-NOMA-rate}
R_{m,k}^{D,NOMA}=\log_2\left(1+\frac{|h_m|^2 \frac{\rho_D}{M}  \alpha_k}{\sum_{j=k+1}^{M}|h_m|^2 \frac{\rho_D}{M} \alpha_j+1}\right),
\end{align}
where $\rho_D=\frac{P_t}{\sigma^2}$.

\subsection{Uplink}
\subsubsection{OMA}
For OMA, the uplink transmission process is similar to that of the downlink. Therefore, the data rate of $\mathrm{U}_m$ is given by
\begin{equation}\label{Eq:systemmodel-uplink-OMA-rate}
\begin{split}
R_m^{U,OMA} = \frac{1}{M}\log_2\left(1+\frac{\eta \rho_U}{|\mathbf{\psi}_m^{Pin}-\mathbf{\psi}_m|^2}\right).
\end{split}
\end{equation}
Here, $\rho_U=\frac{P_u}{\sigma^2}$ and $P_u$ is the total transmit power of the user in the uplink.

\subsubsection{NOMA}
In the uplink system, we assume that no power control is applied for users. Thus, the transmit power of each user is fixed at $P_u$. The BS receives the superimposed signal from $M$ users, which is given by
\begin{equation}\label{systemmodel-uplink-OMA-signal}
\begin{split}
y=\sum_{j=1}^{M} \sqrt[]{P_u} |h_j|^2x_j+\sigma^2,
\end{split}
\end{equation}

According to the principle of successive interference cancellation (SIC), the user $\mathrm{U}_M$, who statistically has a stronger channel, is decoded first~\cite{extract:02-NOMA}. By successively subtracting the decoded signals, the weaker users (e.g., $\mathrm{U}_{M-1}$) are decoded in sequence. This process continues iteratively until all users' signals are successfully decoded. Therefore, the achievable rate of $\mathrm{U}_m$ is given by
\begin{equation}\label{Eq:systemmodel-uplink-NOMA-rate}
\begin{split}
R_{m}^{U,NOMA}=\log_2\left(1+\frac{|h_m|^2}{\sum_{j=1}^{m-1}|h_j|^2+\frac{1}{\rho_U} }\right),
\end{split}
\end{equation}
where $\rho_U=\frac{P_u}{\sigma^2}$, and for $m=1$, we have $\sum_{j=1}^{m-1}|h_j|^2=0$.

\section{Performance Analysis of Downlink Transmission}
\subsection{OMA}
In the OMA scenario, the distance statistics of the users are first characterized, followed by derivation of closed-form expressions for the OP and ER.

\subsubsection{Distance Statistics}
Based on \cite{PAS:07}, we can obtain the following lemma.
\begin{lemma}\label{lemma:distance-statistics}
Let $Z_1 = | \psi_{m}^{Pin} - \psi_m |^{2}$, where $\psi_m^{Pin} = [x_m, 0, d]$, $\psi_m = [x_m, y_m, 0]$, $\forall m$. Its PDF and CDF are given by
\begin{equation}\label{Eq:downlink-OMA-PDF}
\begin{split}
f_{Z_1}(z) = 
\begin{cases} 
0, &z < d^2, \\ 
\frac{1}{D\sqrt{z-d^2}}, & d^2 \leq z < d^2 + \frac{D^2}{4} ,\\ 
0, & d^2 + \frac{D^2}{4} \leq z, 
\end{cases}
\end{split}
\end{equation}
and
\begin{equation}\label{Eq:downlink-OMA-CDF}
\begin{split}
F_{Z_1}(z) = 
\begin{cases} 
0, &z < d^2, \\ 
\frac{2\sqrt{z-d^2}}{D}, & d^2 \leq z < d^2 + \frac{D^2}{4}, \\ 
1, & d^2 + \frac{D^2}{4} \leq z, 
\end{cases}
\end{split}
\end{equation}
respectively.
\end{lemma}

\subsubsection{Outage Probability}
Based on whether the $\mathrm{U}_m$ achieves the target rate within a given time slot, the OP is given by
\begin{equation}\label{Eq:downlink-OMA-OP}
\begin{split}
\mathbb{P}_m^{D,OMA}=\mathrm{P_r}(R_m^{D,OMA}<\tilde R_m) ,
\end{split}
\end{equation}
where $\tilde R_m$ is the target rate of $\mathrm{U}_m$. Therefore, the OP of $\mathrm{U}_m$ can be derived according to the following Theorem \ref{Throrem:downlink-OMA-OP}.
\begin{theorem}\label{Throrem:downlink-OMA-OP}
The OP of $\textrm{U}_m$ is given by
\end{theorem}
\begin{equation}\label{Eq:downlink-OMA-OP-close}
\mathbb{P}_m^{D,OMA}=
\begin{cases}
1, & \beta_m<d^2, \\
1-F_{Z_1}(\beta_m) , & d^2 \le \beta_m < d^2+\frac{D^2}{4}, \\
0, & d^2+\frac{D^2}{4} \le \beta_m,
\end{cases}
\end{equation}
where $\beta_m=\frac{\eta \rho_D}{(2^{M \tilde{R}_m}-1)}$. 

\begin{proof}
    Based on \eqref{Eq:model-downlink-OMA-rate}, we can derive $\mathbb{P}_m^{D,OMA}=\mathrm{P_r}(\frac{1}{M}\log_2(1+\frac{\eta \rho_D}{Z_1})<\tilde R_m)=1-\mathrm{P_r}(Z_1 \le \frac{\eta \rho_D}{(2^{M\tilde R_m}-1)})$. 
    According to \eqref{Eq:downlink-OMA-CDF}, we have $\mathbb{P}_m^{D,OMA}= 1$ for $\beta_m \in (0, d^2)$, $\mathbb{P}_m^{D,OMA}=1-F_{Z_1}(\beta_m)$ for $\beta_m \in [d^2, d^2+\frac{D^2}{4} )$, and $\mathbb{P}_m^{D,OMA}=0$ for $\beta_m \in [d^2+\frac{D^2}{4},+\infty)$. 
    This completes the proof.
\end{proof}

Next, the following proposition for the high-SNR approximation of OP can be obtained.
\begin{proposition}\label{Proposition:downlink-OMA-OP-highSNR}
In the high-SNR regime, $\mathbb{P}_m^{D,OMA}$ can be given by $\mathbb{P}_m^{D,OMA}=0$, where $\mathbb{P}_m^{D,OMA}$ becomes zero for $\rho_D \ge \frac{(d^2+\frac{D^2}{4})(2^{M\tilde R_m} -1)}{\eta}$.
\end{proposition}
\begin{proof}
    When $\rho_D \to \infty$, we have $\beta_m \to \infty$. Therefore, as $F_{Z_{1}}(\infty) = 1$, we have $\mathbb{P}_{m}^{D,OMA} = 0$. The proof is complete. 
\end{proof}

\subsubsection{Ergodic Rate}
Based on \eqref{Eq:model-downlink-OMA-rate}, the following theorem for the ER of $\textrm{U}_m$ can be obtained.
\begin{theorem}\label{Theorem:downlink-OMA-ER}
The ER of $\textrm{U}_m$ is given by
\begin{equation}\label{Eq:downlink-OMA-ER}
\begin{split}
\mathbb{E}\left(R_m^{D,OMA} \right )=\frac{1}{M} \left ( \Lambda_1 + \Lambda_2 \right ) ,
\end{split} 
\end{equation}
where $\Lambda_1=\log_2 \left ( 1+ \frac{\eta\rho_D }{d^2+\frac{D^2}{4} } \right  )$, $\Lambda_2 = \frac{4}{D\ln2} \left [ \sqrt[]{d^2+\eta \rho_D } \arctan\left ( \frac{D}{2 \ \sqrt[]{d^2+\eta \rho_D } } \right ) -d
\arctan\left ( \frac{D}{2d} \right )  \right ]$.

\begin{proof}
Based on \eqref{Eq:downlink-OMA-PDF} and \eqref{Eq:downlink-OMA-CDF}, the ER of $\mathrm{U}_m$ is given by
$
\mathbb{E}\left(R_m^{D,OMA} \right )=\int_{d^2}^{d^2+\frac{D^2}{4} } \frac{1}{M} \log_2 \left ( 1+ \frac{\eta \rho_D }{z} \right)f_{Z_1}(z) dz
$.
For the above expression, since $F_{Z_1}(z)$ is continuous and differentiable over the interval $\left[d^2,d^2+\frac{D^2}{4}\right]$, by applying the integration by parts for the Riemann–Stieltjes integral, it is given by
\begin{equation}
\begin{split}
\mathbb{E}\left(R_m^{D,OMA} \right) \xlongequal{t=\sqrt[]{Z-d^2} } \frac{1}{M} \left [ \Lambda_1 + \frac{2\eta \rho_D }{D \ln2}  \Gamma_1   \right],
\end{split}
\end{equation}
where $\Lambda_1=\log_2 \left ( 1+ \frac{\eta\rho_D }{d^2+\frac{D^2}{4} } \right  )$, $\Gamma_1 =\int_{0}^{\frac{D}{2} }\frac{2t^2}{t^4+(2d^2+\eta \rho_D )t^2+d^2\eta\rho_D}dt$, which is given by $\Gamma_1 = \frac{2}{\eta \rho_D} \left [ \sqrt[]{d^2+\eta \rho_D } \arctan\left ( \frac{D}{2 \ \sqrt{d^2+\eta \rho_D } } \right ) -d \arctan\left ( \frac{D}{2d} \right )  \right ]$.
The proof is complete.
\end{proof}
\end{theorem}

Based on \eqref{Theorem:downlink-OMA-ER}, the following proposition for the asymptotic ER of $\mathrm{U}_m$ can be obtained.
\begin{proposition}\label{Proposition:downlink-OMA-ER-highSNR}
In the high-SNR regime, the approximation of $\mathbb{E}\left(R_m^{D,OMA} \right)$ is given by
\begin{equation}\label{Eq:downlink-OMA-ER-highSNR}
\begin{split}
\mathbb{E}\left(R_m^{D,OMA} \right)^{\infty} \approx \frac{1}{M}\log_2 \left ( \Lambda_1 + \Lambda_3\right  ),
\end{split}
\end{equation}
where $\Lambda_3=\frac{D}{4\ln2}\left (\frac{D}{2}- d \arctan\left(\frac{D}{2d} \right )\right) $.
\end{proposition}
\begin{proof}
For $\Lambda_1$, as $\rho_D \to \infty$, we have $\Lambda_1 \to \infty$. As for $\Lambda_2$, it suffices to analyze the term $\sqrt{d^2+\eta \rho_D} \arctan\left( \frac{D}{2\sqrt{d^2+\eta \rho_D}} \right)$ in the high-SNR regime. To simplify the analysis, let $x = \sqrt{d^2+\eta \rho_D}$, and apply the Taylor expansion at $x \to \infty$ to obtain its approximation as $\underset{x \to \infty}{\lim}x \arctan\left( \frac{D}{2x} \right)=\frac{D}{2}$.
The proof is complete.
\end{proof}

Therefore, the following corollary for the high-SNR slope of $\mathrm{U}_m$ can be obtained.
\begin{corollary}
In the downlink of the OMA system, the high-SNR slope of $\mathrm{U}_m$ is given by $\mathcal{S}_{d}^{D,\mathrm{OMA}} = \frac{1}{M}$. 
\end{corollary}
\begin{proof}
We have $\mathcal{S}_{m}^{D,OMA}=\underset{\rho_D\to \infty }{\mathrm{lim} }  \frac{\mathbb{E}\left(R_m^{D,OMA} \right)^{\infty}}{d \ \log_2(\rho_D) }  = \frac{1}{M} $. This completes the proof.
\end{proof}

\subsection{NOMA}
In the NOMA scenario, we first characterize the order statistics of channel gains and subsequently derive the closed-form expressions for the OP and the ER.
\subsubsection{Order Statistics of Channel Gains}
Based on \eqref{Eq:downlink-OMA-CDF} and \eqref{Eq:downlink-OMA-PDF}, the PDF and CDF of $Y=|h_m|^2$, representing the $m$-th channel gain, can be given by
\begin{equation}\label{Eq:downlink-NOMA-PDF}
\begin{split}
f_{Y}(y) = 
\begin{cases} 
0, & y < \frac{\eta }{d^2 + \frac{D^2}{4} } , \\ 
\frac{\eta }{ D y^2 \sqrt[]{\frac{\eta }{y} - d^2 } } , & \frac{\eta }{d^2 + \frac{D^2}{4} } \leq y < \frac{\eta}{d^2} 
, \\ 
0, & \frac{\eta}{d^2} \leq y, 
\end{cases}
\end{split}
\end{equation}
and
\begin{equation}\label{Eq:downlink-NOMA-CDF}
\begin{split}
F_{Y}(y) = 
\begin{cases} 
0, & y < \frac{\eta }{d^2 + \frac{D^2}{4} } , \\ 
1-\frac{2}{D} \ \sqrt[]{\frac{\eta}{y} - d^2 }  , & \frac{\eta }{d^2 + \frac{D^2}{4} } \leq y < \frac{\eta}{d^2} 
, \\ 
1, & \frac{\eta}{d^2} \leq y .
\end{cases}
\end{split}
\end{equation}

After sorting the users based on their channel conditions and utilizing order statistics, the PDF and CDF of $\mathrm{U}_m$ can be given by $\displaystyle f_{Y_m}(y)={\frac {M!}{(m-1)!(M-m)!}}f_{Y}(y)F_{Y}^{m-1}(y)(1-F_{Y}(y))^{M-m}$ and ${\displaystyle F_{Y_m}(y)=\sum _{j=m}^{M}{\binom {M}{j}}F_{Y}^{j}(y)(1-F_{Y}(y))^{M-j}}$, respectively \cite{casella2002statistical,extract:04-NOMA}.

\subsubsection{Outage Probability}
In the downlink of NOMA, the signals of users with weaker channel conditions are first decoded via SIC, followed by the decoding of users with stronger channels. Therefore, for $\mathrm{U}_m$, the signals of all users with weaker channels must be successfully decoded and removed before its own signal can be correctly decoded.
Hence, the OP of $\mathrm{U}_m$ is given by $\mathbb{P}_m^{D,NOMA}=1-\mathrm{P_r}\left(R_{m,1}^{D,NOMA}\ge \tilde{R}_1, \cdots, R_{m,m}^{D,NOMA}\ge \tilde{R}_m \right )$, where $R_{m,k}^{D,NOMA}$ $(k\le m)$ is the data rate for $\textrm{U}_m$ to decode $\textrm{U}_k$'s signals, $\tilde{R}_k$ is the target rate of $\textrm{U}_k$. 
It is worth noting that the condition $\tilde{R}_k < R_{m,k}^{D,\mathrm{NOMA},\infty}$ always holds, where $R_{m,k}^{D,\mathrm{NOMA},\infty} = \log_2 \left( 1+ \frac{\alpha_k}{\sum_{j=k+1}^{M}\alpha_j}\right)$ is the asymptotic rate in the high-SNR regime. 
Then, the following theorem for the OP of $\textrm{U}_m$ can be obtained.
\begin{theorem}\label{Theorem:downlink-NOMA-OP}
The OP of $\textrm{U}_m$ is given by
\begin{equation}\label{Eq:downlink-NOMA-OP}
\begin{split}
\mathbb{P}_m^{D,NOMA} = 
\begin{cases}
0, & \lambda_m< \frac{\eta }{d^2 + \frac{D^2}{4} }, \\
F_{Y_m}(\lambda_m) , & \frac{\eta }{d^2 + \frac{D^2}{4} } \le \lambda_m < \frac{\eta}{d^2}, \\
1, & \frac{\eta}{d^2} \le \lambda_m,
\end{cases}
\end{split}
\end{equation}
where $\lambda_m=\max {\{\epsilon_1,\cdots, \epsilon_m\}}$, $\epsilon_m=\frac{M }{\rho_D(\frac{\alpha_m}{2^{\tilde{R}_m}-1} -  {\textstyle \sum_{j=m+1}^{M} \alpha_j}   )}$.

\begin{proof}
Based on \eqref{Eq:systemmodel-downlink-NOMA-rate}, the OP of $\mathrm{U}_m$ decoding the signal of $\mathrm{U}_k$ is given by 
$
\mathrm{P_r}\left (R_{m,k}^{D,NOMA} < \tilde{R}_k \right )
=1-\mathrm{P_r}\left ( Y \ge \epsilon_k \right )
$.
According to the principle of successive interference cancellation (SIC), $\mathrm{U}_m$ first decodes $\mathrm{U}_k$’s signal and then decodes its own signal.
Therefore, $\mathrm{U}_k$ must meet the threshold $\tilde{R}_k$ for successful decoding. Then, the OP of $\mathrm{U}_m$ is given by 
\begin{equation}\label{}
\begin{split}
\mathbb{P}_m^{D,NOMA}
=&\Pr \left ( Y < \max {\{\epsilon_1,\cdots, \epsilon_m\}} \right )=F_{Y_m}(\lambda_m).
\end{split}
\end{equation}
This completes the proof.
\end{proof}
\end{theorem}

Therefore, the following proposition for the high-SNR approximation of OP can be obtained.
\begin{proposition}\label{Proposition:downlink-NOMA-OP-highSNR}
In the high-SNR regime, $\mathbb{P}_m^{D,NOMA}$ is given by $\mathbb{P}_m^{D,NOMA} = 0$, where $\mathbb{P}_m^{D,NOMA}$ becomes zero for $\rho_D \ge \frac{d^2+\frac{D^2}{4} }{\eta} \max \left \{ \gamma_1,\dots, \gamma_m  \right \} $, where $\gamma_m=\frac{M}{\frac{\alpha_m}{2^{\tilde{R}_m}-1} -  {\textstyle \sum_{j=m+1}^{M} \alpha_j}}$.

\begin{proof}
    When $\rho_D \rightarrow \infty$, we have $\lambda_m \rightarrow 0$. Therefore, $F_{Y_m}(\lambda_m)=0$. This completes the proof.
\end{proof}
\end{proposition}

\subsubsection{Ergodic Rate}
Based on \eqref{Eq:systemmodel-downlink-NOMA-rate}, the following theorem for the ER of $\textrm{U}_m$ can be obtained.
\begin{theorem}\label{Theorem:downlink-NOMA-ER}
Based on Jensen's inequality, the upper and lower bounds of the ER of $\textrm{U}_m$ can be derived as
\begin{equation}\label{Eq:downlink-NOMA-ER-upper}
\begin{split}
\mathbb{E}\left( R_{m,m}^{D,NOMA}\right)^{u}=\log_2\left ( 1+   \Lambda_3 \Gamma_2  \right ),
\end{split}
\end{equation}
and
\begin{equation}\label{Eq:downlink-NOMA-ER-lower}
\begin{split}
\mathbb{E}\left(R_{m,m}^{D,NOMA}\right)^{l}=\log_2\left ( 1+  \frac{\rho_D \alpha _m}{\rho_D \sum_{ j=m+1 }^{M}\alpha_j + M\Lambda_4 \Gamma_3}    \right ),
\end{split}
\end{equation}
respectively, where $\Lambda_3=\eta  \rho_D \alpha_m\begin{pmatrix}M-1\\m-1\end{pmatrix}  \left ( \frac{2}{D}  \right )^{M-m+1}$, $\Lambda_4=\frac{M}{\eta}\begin{pmatrix}M-1\\m-1\end{pmatrix}\left ( \frac{2}{D}  \right )^{M-m+1}$, $\Gamma_2=\sum_{k=0}^{m-1} \begin{pmatrix}m-1\\k\end{pmatrix}\left (- \frac{2}{D}  \right )^{k}\xi_k$, $\Gamma_3=\sum_{k=0}^{m-1}\begin{pmatrix}m-1\\k\end{pmatrix} \left ( -\frac{2}{D}  \right )^k \left ( \frac{\left ( \frac{D}{2}  \right )^{M-m+k+3} }{M-m+k+3}
+\frac{\left ( \frac{D}{2}  \right )^{M-m+k+1} d^2 }{M-m+k+1}   \right ) $. The value of $\xi_k$ can be found in the proof.

\begin{proof}\label{proof-NOMA-downlink-ER}
See Appendix \ref{Appen-theorem-1}.
\end{proof}
\end{theorem}

Based on \eqref{Theorem:downlink-NOMA-ER}, the following proposition for the asymptotic ER of the users can be obtained.
\begin{proposition}
In the high-SNR regime, the upper and lower bounds of the ER for $\mathrm{U}_M$ can be approximated by
\begin{equation}\label{NOMA-high-SNR-upper}
\begin{split}
\mathbb{E}\left( R_{M,M}^{D,NOMA} \right)^{u,\infty}=\log_2\left ( 1+   \Lambda_3 \Gamma_2  \right ),
\end{split}
\end{equation}
and
\begin{equation}\label{NOMA-high-SNR-lower}
\begin{split}
\mathbb{E}\left(R_{M,M}^{D,NOMA} \right)^{l,\infty}=\log_2\left ( 1+  \frac{\rho_D \alpha _m}{M\Lambda_4 \Gamma_3}    \right ),
\end{split}
\end{equation}
respectively, where $m=M$. 
Furthermore, the upper and lower bounds of the ER for $\mathrm{U}_m$ $(m \ne M)$ can be approximated as
\begin{equation}
\begin{split}
\mathbb{E}\left( R_{m,m}^{D,NOMA} \right)^{u,\infty}=\log_2\left ( 1+  \frac{ \alpha _m \Lambda_5}{ \sum_{ j=m+1 }^{M}\alpha_j } \right),
\end{split}
\end{equation}
and
\begin{equation}
\begin{split}
\mathbb{E}\left( R_{m,m}^{D,NOMA} \right)^{l,\infty}=\log_2\left ( 1+  \frac{ \alpha _m}{ \sum_{ j=m+1 }^{M}\alpha_j }  \right ),
\end{split}
\end{equation}
where $\Lambda_5=M\begin{pmatrix}M-1\\m-1\end{pmatrix}  \sum_{k=0}^{m-1} \begin{pmatrix}m-1\\k\end{pmatrix}\frac{(-1)^{k}}{M-m+k+1}$.
\end{proposition}

\begin{proof}
For the last user to decode, i.e., $\mathrm{U}_M $, all other users' signals have already been decoded and subtracted. 
As a result, only its own signal remains. 
Therefore, for the upper bound of the ER, when \( m = M \), we have $\underset{\rho_D \to \infty}{\mathrm{lim}} R_{M,M}^{D,NOMA,u}=\log_2\left ( 1+   \Lambda_3 \Gamma_2  \right )$. 
Similarly, for the lower bound of the ER, we obtain $\underset{\rho_D \to \infty}{\mathrm{lim}} R_{M,M}^{D,NOMA,l}=\log_2\left ( 1+  \frac{\rho_D \alpha _m}{M\Lambda_4 \Gamma_3} \right )$.

According to Appendix \ref{Appen-theorem-1}, the upper bound of ER for $\mathrm{U}_m$ is given by $\underset{\rho_D \to \infty}{\mathrm{lim}}R_m^{D,NOMA,u}=\underset{\rho_D \to \infty}{\mathrm{lim}}\log_2\left ( 1+\mathbb{E}\left (f_2(y)\right ) \right )$. 
As $ \rho_D \to \infty $, the expectation $ \mathbb{E}\left(f_2(y)\right) $ is given by
$
\underset{\rho_D \to \infty}{\mathrm{lim}}\mathbb{E}\left (f_2(y)\right )= \underset{\rho_D \to \infty}{\mathrm{lim}}\eta \alpha_m\begin{pmatrix}M-1\\m-1\end{pmatrix}  \left ( \frac{2}{D}  \right )^{M-m+1} \\
\times\int_{0}^{\frac{D}{2} } \frac{t^{M-m}\left ( 1-\frac{2}{D}t  \right )^{m-1} }
{ \frac{t^2 + d^2}{\rho_D} + \frac{\eta   {\textstyle \sum_{j=m+1}^{M}}\alpha_j }{M} } dt.
$
where $\frac{t^2 + d^2}{\rho_D} \to 0$, Therefore, we have $\underset{\rho_D \to \infty}{\mathrm{lim}}\mathbb{E}\left (f_2(y)\right )=  \frac{ \alpha _m \Lambda_4}{ \sum_{ j=m+1 }^{M}\alpha_j }$. 
\end{proof}

Therefore, the following corollary for the high-SNR slope of users' ER can be obtained.
\begin{corollary}
In the downlink of the NOMA system, the high-SNR slope of the ER for $\mathrm{U}_M$ and $\mathrm{U}_m$ $(m \ne M)$ is given by $\mathcal{S}_{M}^{D,NOMA} = 1$ and $\mathcal{S}_{m}^{D,NOMA} = 0$, respectively. 
\end{corollary}
\begin{proof}
Based on \eqref{NOMA-high-SNR-upper} and \eqref{NOMA-high-SNR-lower}, the high-SNR slope of $\mathrm{U}_M$ is obtained as $\mathcal{S}_{M}^{D,NOMA}=\underset{\rho_D\to \infty }{\mathrm{lim} }  \frac{d \ \mathbb{E}\left( R_{M,M}^{D,NOMA} \right)^{l,\infty}}{d \ \log_2(\rho_D) }  = 1 $. Since in the high SNR regime, the ER at user $ \mathrm{U}_m $ $(m \ne M)$ tends to a constant and no longer varies, we have $\mathcal{S}_{m}^{D,NOMA} =\underset{\rho_D\to \infty }{\mathrm{lim} }  \frac{d \ \mathbb{E}\left( R_{m,m}^{D,NOMA} \right)^{l,\infty}}{d \ \log_2(\rho_D) }  = 0 $. This completes the proof.
\end{proof}

\subsection{A Discussion on Outage Performance Comparison Between OMA and NOMA}\label{SubSection:downlink-NOMA-OP-compare}
In the downlink scenario, it is worth noting that the outage performance of OMA and NOMA is not always identical under the same target rate. 
Therefore, we compare the reliability of the two schemes by the minimum SNR threshold at which the OP vanishes.
Here, we define the minimum SNR required for $\mathrm{U}_m$ to achieve zero OP in OMA and NOMA systems as $\rho_m^{D,OMA}=\frac{(d^2+\frac{D^2}{4})(2^{M\tilde R_m} -1)}{\eta}$ and $\rho_m^{D,NOMA}=\frac{d^2+\frac{D^2}{4} }{\eta} \max \left \{ \gamma_1,\dots, \gamma_m \right \}$, respectively. 
Therefore, the discussion can be divided into the following three cases.

\textit{Case 1}: $\rho_m^{D,OMA} < \underset{\substack{m=1,2,\dots,M}}{\min}\rho_m^{D,NOMA}$, which can be expressed as follows
\begin{equation}\label{}
\begin{split}
2^{M\tilde R_m} -1 < \gamma_1,
\end{split}
\end{equation}
where $\gamma_1=\frac{M}{\frac{\alpha_1}{2^{\tilde{R}_1}-1} -  1+\alpha_1}$, $\tilde{R}_1=\tilde{R}_m$. In this case, under the same target rate, the SNR required for users in the OMA system to achieve an OP of zero is lower than that for any user in the NOMA system. 

\textit{Case 2:} $\rho_m^{D,OMA} \ge \rho_m^{D,NOMA}$, which can be expressed as follows
\begin{equation}\label{Eq:downlink-NOMA-OP-remark-case2}
\begin{split}
2^{M\tilde R_m} -1 \ge \max \left \{ \gamma_1,\dots, \gamma_m  \right \}.
\end{split}
\end{equation}
Conversely, in this case, the SNR required for users in the OMA system to achieve an OP of zero may be higher than that for some users in the NOMA system.

\textit{Case 3}: $\rho_m^{D,OMA} \ge \underset{\substack{m=1,2,\dots,M}}{\max}\rho_m^{D,NOMA}$. 
In this case, the OP performance of OMA would be inferior to that of all NOMA users. However, such a scenario does not occur in practice.

\begin{proof}
For case 1, $\underset{\substack{m=1,2,\dots,M}}{\min}\rho_m^{D,NOMA}$ can be given by $\min \left \{ \gamma_1,\dots, \max \left \{ \gamma_1,\dots, \gamma_M \right \} \right \}$. Since we have $\underset{\substack{m=1,2,\dots,M}}{\min}\rho_m^{D,NOMA}=\gamma_1$, $\forall \gamma_1,\dots, \gamma_M$.

For case 2, we can obtain \eqref{Eq:downlink-NOMA-OP-remark-case2} by simplification. 

For case 3, in order to ensure that no user in different systems experiences a perpetual outage, we require 
$\tilde R_m \le R_m^{\infty}$, where $R_m^{\infty}=\log_2 \left (1+ \frac{\alpha_m}{\sum_{j=m+1}^{M}\alpha_j} \right )$ denotes the asymptotic rate in the high-SNR regime. Without loss of generality, let $\tilde R_m = R_{\min}^{\infty}=\min\{ R_1^{\infty}, R_2^{\infty}, \dots ,R_{M-1}^{\infty}\}$. Then we have
\begin{equation}\label{}
\begin{split}
2^{MR_{\min}^{\infty}} -1 \ge \underset{m=1,2,\dots,M}{\max} \frac{M}{\frac{\alpha_m}{2^{R_{\min}^{\infty}}-1} -\sum_{j=m+1}^{M}\alpha_j }.
\end{split}
\end{equation}
Suppose there exists a user $\mathrm{U}_k$ such that $R_{\min,k}^{\infty}= \log_2 \left (1+ \frac{\alpha_k}{\sum_{j=k+1}^{M}\alpha_j} \right )$ is the minimum target rate. 
In this case, $\rho_k^{D,NOMA} \to \infty$. Consequently, $\underset{m=1,2,\dots,M}{\max}\rho_m^{D,NOMA} = \infty$, and the inequality does not hold. 
This completes the proof.
\end{proof}

\section{Performance Analysis of Uplink Transmission}
\subsection{OMA}
In the uplink OMA system, the performance analysis of users is similar to that of the downlink, owing to the inherent symmetry in the communication mechanisms. 
By replacing the transmit power of the BS in the downlink with the transmit power of users in the uplink, the corresponding OP and ER can be derived. 
The detailed results will be provided in subsequent sections.

\subsection{NOMA}
In the uplink, the performance analysis of the NOMA system is similar to that of the downlink.
The users are ordered according to their channel conditions in ascending order. 
Without loss of generality, we assume that $\mathrm{U}_m$ is the $m$th user, and $|h_1|^2\leq\cdots\leq|h_M|^2$, where $|h_m|^2 = \frac{\eta}{|\mathbf{\psi}_m^{Pin}-\mathbf{\psi}_m|^2}$.
\subsubsection{Outage Probability}
For \( \mathrm{U}_m \), successful decoding of its signal requires that all preceding users \( (\mathrm{U}_{m+1}, \dots, \mathrm{U}_M) \) are successfully decoded. 
Therefore, according to \eqref{Eq:systemmodel-uplink-NOMA-rate}, the OP of \( \mathrm{U}_m \) is given by
$\mathbb{P}_m^{U,NOMA} = 1 - \Pr \left( R_m^{U,NOMA} \ge \tilde{R}_m, \cdots \right. \\
\left. R_M^{U,NOMA} \ge \tilde{R}_M \right)$. For $ R_m^{U,NOMA} \ge \tilde{R}_m$, we have $|h_m|^2 \ge \left ( \sum_{j=1}^{m-1}|h_j|^2+\frac{1}{ \rho_U}  \right ) \left ( 2^{\tilde{R}_m} -1\right )$. Then, by substituting the above inequality into $\mathbb{P}_m^{U,NOMA}$, the OP of \( \mathrm{U}_m \) in the uplink NOMA system can be given by
\begin{equation}\label{Eq:uplink-NOMA-OP}
\begin{split}
 \mathbb{P}_m^{U,NOMA} =& 1 - \underbrace{\int_{a}^{b} \cdots  \int_{a}^{b}}_{m-1 \ \mathrm{integrals} } 
\underbrace{\int_{\omega_m}^{b} \cdots  \int_{\omega_M}^{b}}_{M-m+1 \ \mathrm{integrals} } 
f_{|h_1|^2|h_2|^2 \cdots |h_M|^2} \\
&(y_1,y_2,\dots, y_M  ) dy_1 dy_2 \cdots dy_M.
\end{split}
\end{equation}
where $f(\cdot)$ denotes the joint probability density function, $a=\frac{\eta}{d^2+\frac{D^2}{4}}$, $b=\frac{\eta}{d^2}$ and $\omega_i$ can be expressed as
\begin{equation}\label{Uplink:NOMA-OP-omga}
\begin{split}
\omega_i=
\begin{cases}
\left ( \sum_{j=1}^{i-1}|h_j|^2+\frac{1}{ \rho_U}  \right ) \left ( 2^{\tilde{R}_i} -1\right ),  &1<i\le M, \\
\frac{2^{\tilde{R}_1} -1}{\rho_U}, &i=1.
\end{cases}
\end{split}
\end{equation}

\paragraph{\textbf{The special case of two users}}
To validate the system performance, we derive the OPs of the two-user case. The OPs for the two users can be respectively expressed as 
\begin{equation}\label{eq:outage_U1_uplink_NOMA}
\begin{split}
\mathbb{P}_1^{U,NOMA}=1-\Pr\left(R_1^{U,NOMA} \ge \tilde{R}_1, R_2^{U,NOMA} \ge \tilde{R}_2 \right) ,
\end{split}
\end{equation}
and
\begin{equation}\label{eq:outage_U2_uplink_NOMA}
\begin{split}
\mathbb{P}_2^{U,NOMA}=\Pr(R_2^{U,NOMA} < \tilde{R}_2).
\end{split}
\end{equation}

Then, the following theorem for the OPs of $\mathrm{U}_1$ and $\mathrm{U}_2$ can be obtained.
\begin{theorem}\label{OP:uplink-NOMA}
According to \eqref{Eq:systemmodel-uplink-NOMA-rate}, \eqref{eq:outage_U1_uplink_NOMA} and \eqref{eq:outage_U2_uplink_NOMA}, the OP of $\mathrm{U}_1$ is given by
\begin{equation}\label{}
\begin{split}
\mathbb{P}_1^{U,NOMA} \approx \left\{\begin{matrix}
  1, & \omega_1 \ge b,\\
  1 - \left (I_1 + I_2 \right), & 0 < \tilde{R}_2 < 1 ,\omega_1 < b,\\
  1 - I_3,  & \tilde{R}_2=1,\omega_1 < b,\\
  1 - \left (I_4 + I_5 \right), & 1 < \tilde{R}_2,\omega_1 < b,\\
\end{matrix}\right.
\end{split}
\end{equation}
where $\omega_1=\frac{2^{\tilde{R}_1} -1}{\rho_U}$, $\omega_2(x_i)=\left ( x_i+\frac{1}{ \rho_U}  \right ) \left ( 2^{\tilde{R}_2} -1\right )$.
The parameters are given by 
$c=\frac{2^{\tilde R_2}-1}{\rho_U (2-2^{\tilde R_2})}$, 
$w_i=\frac{\pi}{n}$, 
$t_i=\mathrm{cos}\left( \frac{2i-1}{2n} \pi\right)$.
The function $f(x_i)$ is given by 
$f\left( x_i \right)=\frac{\sqrt[]{\frac{\eta}{\min(\omega_2(x_i),b)}-d^2 } }{x_i^2\sqrt[]{\frac{\eta}{x_i}-d^2 } }$.
The components $I_1$ through $I_5$ are expressed as:
$I_1 = \frac{4\eta}{D^2} \left ( \frac{1}{\max(\omega_1, c, a)} - \frac{1}{b}  \right )$;
$I_2 \approx  \frac{2\eta \left (  \min(b,c)-\max(\omega_1, a)\right ) }{D^2} \times \sum_{i=1}^{n} w_i \ \sqrt[]{1-t_i^2} f\left ( x_i \right )$, where $x_i=\frac{\min(b,c)-\max(\omega_1, a)}{2}t_i+\frac{\min(b,c)+\max(\omega_1, a)}{2}$;
$I_3 \approx \frac{2\eta \left (  b-\max(\omega_1, a)\right ) }{D^2} \sum_{i=1}^{n} w_i \ \sqrt[]{1-t_i^2} f\left ( x_i \right ) $, where $x_i=\frac{b-\max(\omega_1, a)}{2}t_i+\frac{b+\max(\omega_1, a)}{2}$;
$I_4 = \frac{4\eta}{D^2} \left ( \frac{1}{\max(\omega_1, a)} - \frac{1}{\min(b, c)}  \right )$; 
$I_5 \approx \frac{2\eta \left (  b-\max(\omega_1, c, a)\right ) }{D^2} \sum_{i=1}^{n} w_i \ \sqrt[]{1-t_i^2} f\left ( x_i \right )$, 
where $x_i=\frac{b-\max(\omega_1, c, a)}{2}t_i+\frac{b+\max(\omega_1, c, a)}{2}$. The OP expression for $\mathrm{U}_2$ can be given by
\begin{equation}\label{}
\begin{split}
\mathbb{P}_2^{U,NOMA} \approx 
\begin{cases}
I_6, & 0 < \tilde{R}_2 < 1,\ c>a,  \\
I_7, & \tilde{R}_2=0,  \\
I_8, & 1 < \tilde{R}_2 ,\ c<b, 
\end{cases}
\end{split}
\end{equation}
where the components $I_6$ through $I_8$ are expressed as: $I_6=\frac{4\eta }{D^2} \left (  \frac{1}{a}-\frac{1}{\min(b,c)} \right )  -  \frac{2\eta \left (  \min(b,c)-a\right ) }{D^2} \sum_{i=1}^{n} w_i \ \sqrt[]{1-t_i^2} f\left ( x_i \right )$, 
$x_i=\frac{\min(b,c)-a}{2}t_i+\frac{\min(b,c)+a}{2}$;
$I_7=1 -  \frac{2\eta \left (  b-a\right ) }{D^2} \sum_{i=1}^{n} w_i \ \sqrt[]{1-t_i^2} f\left ( x_i \right ) $, 
$x_i=\frac{b-a}{2}t_i+\frac{b+a}{2}$;
$I_8=\frac{4\eta }{D^2} \left (  \frac{1}{\max(a,c)}-\frac{1}{b} \right ) -  \frac{2\eta \left (  b-\max(a,c)\right ) }{D^2} \sum_{i=1}^{n} w_i \ \sqrt[]{1-t_i^2} f\left ( x_i \right )$, 
$x_i=\frac{b-\max(a,c)}{2}t_i+\frac{b+\max(a,c)}{2}$. 

\begin{proof}
    See Appendix \ref{Appendix:OP-uplink-NOMA}.
\end{proof}
\end{theorem}

Then, the asymptotic OP of the two users in the high-SNR regime can be characterized by the following proposition.
\begin{proposition}\label{Proposition:uplink-NOMA-OP-infty}
In the high-SNR regime, $\mathbb{P}_1^{U,NOMA}$ and $\mathbb{P}_2^{U,NOMA}$ can be expressed as
\begin{equation}\label{}
\begin{split}
\mathbb{P}_1^{U,NOMA,\infty}=\mathbb{P}_2^{U,NOMA,\infty}  = \left\{\begin{matrix}
  0, & \tilde{R}_2 \in (0,1),\\
  \frac{2d^4}{D^2\eta}\frac{\ln \rho_U}{\rho_U}, & \tilde{R}_2=1,\\
  1-I_9, & \tilde{R}_2 \in (1,\infty).
\end{matrix}\right.
\end{split}
\end{equation}
For $\tilde{R}_2 \in (0,1)$, $\mathbb{P}_{\{ 1,2\}}^{U,NOMA}$ becomes zero for $\rho_U \ge \frac{2^{\tilde R_2}-1}{a (2-2^{\tilde R_2})}$. 
For $\tilde{R}_2=1$, $\mathbb{P}_{\{ 1,2\}}^{U,NOMA}$ decays as $\rho_{U}^{-1}$ in the high-SNR regime.
For $\tilde{R}_2 \in (1,\infty)$, $\mathbb{P}_1^{U,NOMA}$ and $\mathbb{P}_2^{U,NOMA} $ converge to a constant value. we have $I_9= \frac{2\eta (b-a ) }{D^2} \sum_{i=1}^{n} w_i \ \sqrt[]{1-t_i^2} f\left ( x_i \right ) $, where $f(x_i)=\frac{\sqrt[]{\frac{\eta}{\min \left ( x_i\left ( 2^{\tilde{R}_2}-1,b\right)\right) }-d^2 } }{x_i^2 \sqrt[]{\frac{\eta}{x_i}-d^2 } }$, $x_i=\frac{b-a}{2}t_i+\frac{b+a}{2}$.
\begin{proof}
    See Appendix \ref{Appendix:OP-uplink-NOMA-infty}.
\end{proof}
\end{proposition}

\subsubsection{Egodic Rate}
Based on \eqref{Eq:systemmodel-uplink-NOMA-rate}, the ER of $\mathrm{U}_m$ is given by
\begin{equation}\label{ER-uplink-NOMA}
\begin{split}
\mathbb{E}\left ( R_m^{U,NOMA} \right ) =\mathbb{E}\left ( \log_2\left(1+\frac{|h_m|^2}{\sum_{j=1}^{m-1}|h_j|^2+\frac{1}{\rho_U} }\right)   \right ) .
\end{split}
\end{equation}
Specifically, when $m=1$, we have $\sum_{j=1}^{m-1}|h_j|^2=0$. Hence, based on the distribution of the channel and the rate expression of the statistically ordered $\mathrm{U}_m$, the ER of $\mathrm{U}_m$ in the uplink NOMA system can be expressed as
\begin{equation}\label{ER-uplink-NOMA-derive}
\begin{split}
\mathbb{E}\left ( R_m^{U,NOMA} \right ) = \underbrace{ \int_{a}^{b} \cdots  \int_{a}^{b}}_{m \ \mathrm{integrals} }
\log_2\left(1+\frac{y_m}{\sum_{j=1}^{m-1}y_j+\frac{1}{\rho_U} }\right) \\
\times f_{|h_1|^2|h_2|^2 \cdots |h_m|^2}
(y_1,y_2,\dots, y_m  ) dy_1 dy_2 \cdots dy_m,
\end{split}
\end{equation}
where $a=\frac{\eta}{d^2+\frac{D^2}{4}}$, $b=\frac{\eta}{d^2}$.

\paragraph{\textbf{The special case of two users}}
For the case of two users, their ERs can be respectively expressed as
\begin{equation}\label{}
\begin{split}
\mathbb{E}\left ( R_1^{U,NOMA} \right ) =\mathbb{E}\left ( \log_2\left(1+|h_1|^2\rho_U \right)   \right ),
\end{split}
\end{equation}
and
\begin{equation}\label{}
\begin{split}
\mathbb{E}\left ( R_2^{U,NOMA} \right ) =\mathbb{E}\left ( \log_2\left(1+\frac{|h_2|^2}{|h_1|^2+\frac{1}{\rho_U} }\right)   \right ) .
\end{split}
\end{equation}

Then, their explicit expressions are further derived and presented in the following theorem.
\begin{theorem}\label{theorem-uplink-NOMA-ER-two-User}
In the uplink NOMA system, the ER of $\mathrm{U}_1$ is given by
\begin{equation}\label{}
\begin{split}
\mathbb{E}\left ( R_1^{U,NOMA} \right ) = -\frac{4\eta}{D^2\ln2}\left( G(b) - G(a)\right),
\end{split}
\end{equation}
where $G(x)=\frac{\ln(1+\rho_U x)}{x}-\rho_U \ln x+\rho_U\ln(1+\rho_U x)$.

On the other hand, the ER of $\mathrm{U}_2$ is given by
\begin{equation}\label{}
\begin{split}
&\mathbb{E}\left ( R_2^{U,NOMA} \right )  \approx\frac{b-a}{\ln2} \sum_{i=1}^{n}w_i \sqrt{1-t_i^2} H_1\left( x _i \right) ,
\end{split}
\end{equation}
where $w_i=\frac{\pi}{n}$, $t_i=\mathrm{cos}\left( \frac{2i-1}{2n} \pi\right)$, $x_i=\frac{b-a}{2}t_i+\frac{b+a}{2}$, $H_1\left( x _i \right)=f_Y(x_i)\Xi_2(x_i)$, $\Xi_2(x_i)=\frac{2}{D} \ln\left ( 1+\frac{x_i }{x_i + \frac{1}{\rho_U } }  \right ) \sqrt[]{\frac{\eta}{x_i}-d^2 }
+ \frac{4}{D}\sqrt[]{\frac{\eta}{x_i + \frac{1}{\rho_U } } + d^2 }\arctan \left ( \frac{\sqrt[]{\frac{\eta}{x_i}-d^2 } }{\sqrt[]{\frac{\eta}{x_i + \frac{1}{\rho_U } } + d^2 }}  \right ) -   \frac{4d}{D} \arctan\left ( \frac{\sqrt[]{\frac{\eta}{x_i}-d^2 }}{d}  \right )$, $f_Y(x_i) = \frac{\eta }{ D x_i^2 \sqrt[]{\frac{\eta }{x_i} - d^2 } }$.
\begin{proof}
    See Appendix \ref{Appen-theorem-uplink-NOMA-ER-two-User}.
\end{proof}
\end{theorem}
Then, the asymptotic ERs of the two users in the high-SNR regime are characterized by the following proposition.
\begin{proposition}
In the high-SNR regime, the approximation of $R_1^{U,NOMA}$ is given by
\begin{equation}\label{ER_uplink_NOMA_lim}
\begin{split}
\mathbb{E}\left ( R_1^{U,NOMA} \right )^{\infty} = \frac{4\eta}{D^2\ln2}\left(\left (  \frac{1}{a} -\frac{1}{b} \right ) \ln \rho_U
+ \Delta_2\right),
\end{split}
\end{equation}
where $\Delta_2 =\frac{\ln a}{a}- \frac{\ln b}{b}+\frac{1}{a}-\frac{1}{b}$.
On the other hand, the approximation of $R_2^{U,NOMA}$ is given by
\begin{equation}\label{}
\begin{split}
&\mathbb{E}\left ( R_2^{U,NOMA} \right ) ^{\infty} \approx\frac{b-a}{\ln2} \sum_{i=1}^{n}w_i \sqrt{1-t_i^2} H_2\left( x _i \right) ,
\end{split}
\end{equation}
where $w_i=\frac{\pi}{n}$, $t_i=\mathrm{cos}\left( \frac{2i-1}{2n} \pi\right)$, $x_i=\frac{b-a}{2}t_i+\frac{b+a}{2}$, $H_2\left( x _i \right)=f_Y(x_i)\Xi_4(x_i)$, $\Xi_3(x_i)=\frac{2\ln2}{D}\sqrt[]{\frac{\eta}{x_i}-d^2 } 
+ \frac{4}{D}\sqrt[]{\frac{\eta}{x_i} + d^2 }\arctan \left ( \frac{\sqrt[]{\frac{\eta}{x_i}-d^2 } }{\sqrt[]{\frac{\eta}{x_i} + d^2 }}  \right ) -   \frac{4d}{D} \arctan\left ( \frac{\sqrt[]{\frac{\eta}{x_i}-d^2 }}{d}  \right )$, $f_Y(x_i) = \frac{\eta }{ D x_i^2 \sqrt[]{\frac{\eta }{x_i} - d^2 } }$.

\begin{proof}
For $\mathrm{U}_1$, the term $G(b)-G(a)$ can be simplified to the following form
$
G(b)-G(a)=\frac{\ln (1+\rho_U b)}{b} - \frac{\ln (1+\rho_U a)}{a} + \rho_U \ln \left ( \frac{1+\frac{1}{\rho_U b} }{1+\frac{1}{\rho_U a}}  \right ) .
$
As $\rho_U \to \infty$, using the Taylor series expansion, we obtain $\underset{\rho_U \to \infty}{\lim}G(b)-G(a)=\underset{\rho_U \to \infty}{\lim}\left (  \frac{1}{a} -\frac{1}{b} \right ) \ln \rho_U + \Delta_2$,
where $\Delta_2=\frac{\ln b}{b}-\frac{\ln a}{a} +\frac{1}{b} - \frac{1}{a}$.

For $\mathrm{U}_2$, it is sufficient to consider the asymptotic approximation of $\Xi_2(x_i)$ at high SNR. Therefore, when $\rho_U \to \infty$, we obtain $\underset{\rho_U \to \infty}{\lim}\Xi_2(x_i)=\Xi_4 (x_i)= \frac{2\ln2}{D}\sqrt[]{\frac{\eta}{x_i}-d^2 } 
\\
+ \frac{4}{D}\sqrt[]{\frac{\eta}{x_i} + d^2 }\arctan \left ( \frac{\sqrt[]{\frac{\eta}{x_i}-d^2 } }{\sqrt[]{\frac{\eta}{x_i} + d^2 }}  \right )
-   \frac{4d}{D} \arctan\left ( \frac{\sqrt[]{\frac{\eta}{x_i}-d^2 }}{d}  \right )$.
This completes the proof.
\end{proof}
\end{proposition}

Then, the high-SNR slopes of the users can be characterized by the following corollary.
\begin{corollary}
In the uplink of the NOMA system, the high-SNR slopes of $\mathrm{U}_1$ and $\mathrm{U}_2$ are given by $\mathcal{S}_{1}^{U,NOMA} = \frac{4\eta}{D^2}\left ( \frac{1}{a} - \frac{1}{b}\right)$ and $\mathcal{S}_{2}^{U,NOMA} = 0$. 

\begin{proof}
Based on \eqref{ER_uplink_NOMA_lim}, we can derive $\mathcal{S}_{1}^{U,NOMA}=\underset{\rho_U \to \infty}{\lim}\frac{d \ \mathbb{E}\left ( R_1^{U,NOMA} \right )^{\infty}  }{d \ \log_2 (\rho_U)}=\frac{4\eta}{D^2}\left ( \frac{1}{a} - \frac{1}{b}\right)$ and $\mathcal{S}_{2}^{U,NOMA}=\underset{\rho_U \to \infty}{\lim}\frac{d \ \mathbb{E}\left ( R_2^{U,NOMA} \right )^{\infty}  }{d \ \log_2 (\rho_U)}=0$.
This completes the proof.
\end{proof}
\end{corollary}

\section{Numerical Results}

In this section, computer simulations are employed to evaluate the performance of PA systems with multiple PAs and a single waveguide in both OMA and NOMA scenarios and to verify the correctness of the theoretical derivations. In the simulations, both uplink and downlink scenarios are analyzed, and the effects of different parameters on system performance are compared. Unless otherwise specified, the main simulation parameters are set as shown in Table~\ref{tab:Parameters}. The simulations adopt the Monte Carlo method, executed $1 \times 10^6$ times in MATLAB, with the averaged results taken as the final simulation outcomes to approximate practical scenarios and reduce computational errors.
\begin{table}
\centering
\caption{Parameters setting}
\begin{tabular}{|p{0.24\textwidth}|p{0.16\textwidth}|}
\hline
Side length of the room & $D = 20m$\\
\hline
Height of the waveguide& $d = 5m$\\
\hline
Carrier frequency& $f_{c} = 10\ GHz$\\
\hline
Bandwidth& $B = 1\ MHz$\\
\hline
Speed of light & $c = 3 \times10^8 \ \text{m/s}$\\
\hline
The power allocation coefficients for $\mathrm{U}_1$ and $\mathrm{U}_2$ (NOMA)&  $\alpha_1 = 0.8$ and $\alpha_2 = 0.2$\\
\hline
Target data-rate & $\tilde{R} = 0.5\ Mbps$\\
\hline
Number of points for Chebyshev-Gauss quadratures& $n=100$\\
\hline
\end{tabular}\label{tab:Parameters}
\end{table}

\subsection{Downlink}

\begin{figure}[t]
\centering
\includegraphics[width=\linewidth]{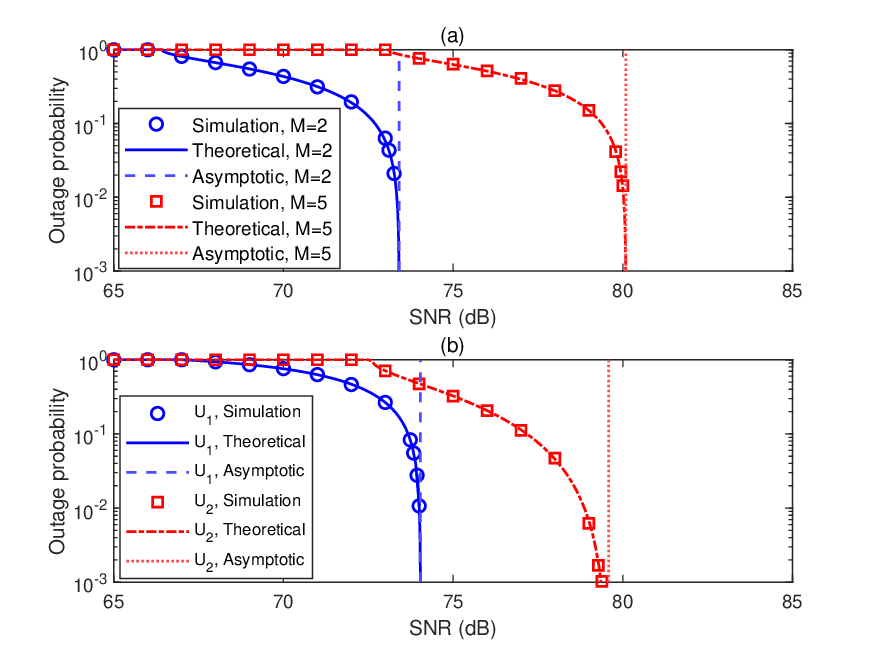}
\caption{Outage probability versus SNR for downlink OMA and NOMA systems. (a) illustrates the OMA system performance for $M=2$ and $M=5$, respectively. (b) illustrates the NOMA system in a two-user scenario.}
\label{fig:OMA-NOMA-Downlink}
\end{figure}
Fig.~\ref{fig:OMA-NOMA-Downlink} illustrates the OP versus the transmit SNR in the downlink for both OMA and NOMA systems. 
As shown in Fig.~\ref{fig:OMA-NOMA-Downlink}(a), in the PA scenario, OMA is employed to serve multiple users, where simulations are conducted for the cases of $M=2$ and $M=5$. It can be observed that the OP of users in both cases converges to zero in the high-SNR regime. However, when $M=2$, the SNR required for users to achieve an OP of $10^{-3}$ is significantly lower than that in the case of $M=5$. The reason is that as the number of users increases, the service time allocated to each user decreases, resulting in degraded outage performance. 
As shown in Fig.~\ref{fig:OMA-NOMA-Downlink}(b), NOMA is considered with two users, i.e., $M=2$. In this case, compared to $\mathrm{U}_2$, $\mathrm{U}_1$ achieves an OP of $10^{-3}$ at a much lower SNR, exhibiting a superior outage performance. This is because $\alpha_1=0.8 \gg \alpha_2=0.2$, meaning that $\mathrm{U}_1$ is allocated more power than $\mathrm{U}_2$. Consequently, even though $\mathrm{U}_1$ experiences a weaker channel condition, its reliability is still better than that of $\mathrm{U}_2$.

\begin{figure}[t]
\centering
\includegraphics[width=\linewidth]{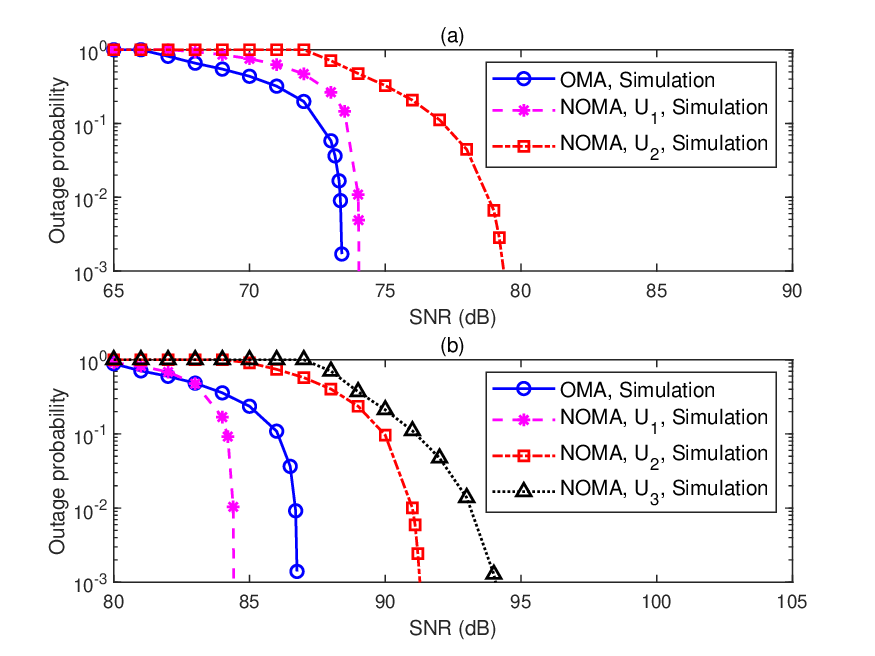}
\caption{Outage probability versus SNR for downlink OMA and NOMA systems. (a) and (b) illustrate case 1 and case 2 in \ref{SubSection:downlink-NOMA-OP-compare}, respectively.}
\label{fig:OMA_NOMA_Integration_OP_Downlink_semilogy}
\end{figure}
Fig.~\ref{fig:OMA_NOMA_Integration_OP_Downlink_semilogy} illustrates the comparison of outage performance between OMA and NOMA systems under different configurations.
For the case of $M=2$, with $\alpha_1=0.8$, $\alpha_2=0.2$, and $\tilde{R}=0.5~\text{Mbps}$, as shown in Fig.~\ref{fig:OMA_NOMA_Integration_OP_Downlink_semilogy}(a), the OMA users outperform all NOMA users in terms of OP performance, which is consistent with case 1 described in \ref{SubSection:downlink-NOMA-OP-compare}.
For the case of $M=3$, with $\alpha_1=0.80$, $\alpha_2=0.16$, $\alpha_3=0.04$, and $\tilde{R}=1.5~\text{Mbps}$, as shown in Fig.~\ref{fig:OMA_NOMA_Integration_OP_Downlink_semilogy}(b), user $\mathrm{U}_1$ in the NOMA system achieves better performance than any OMA user, whereas $\mathrm{U}_2$ and $\mathrm{U}_3$ exhibit inferior performance compared to OMA users, which is consistent with case 2 in \ref{SubSection:downlink-NOMA-OP-compare}.

\begin{figure}[t]
\centering
\includegraphics[width=\linewidth]{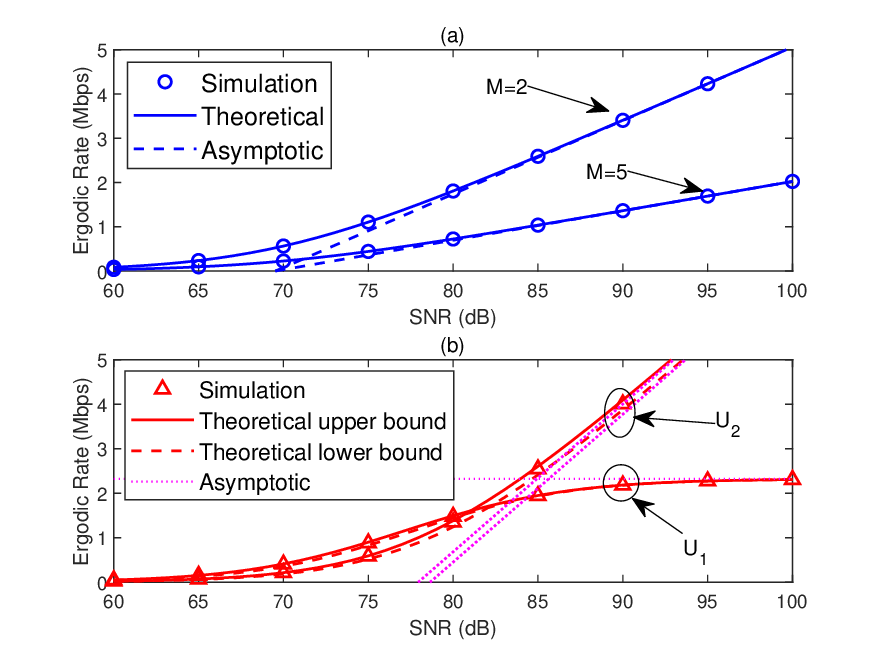}
\caption{Ergodic rate versus SNR in the downlink. (a) illustrates the OMA system performance for $M=2$ and $M=5$, respectively. (b) illustrates the NOMA system in a two-user scenario.}
\label{fig:OMA_NOMA_Rate_Downlink}
\end{figure}
Fig.~\ref{fig:OMA_NOMA_Rate_Downlink} shows ER versus transmit SNR for downlink OMA and NOMA systems. In Fig.~\ref{fig:OMA_NOMA_Rate_Downlink}(a), we simulate two cases with different numbers of users, i.e., $M=2$ and $M=5$.
The ER of each user in the OMA system decreases significantly as the number of users increases. This trend is consistent with the degradation observed in outage performance, as a larger number of users results in less communication resources allocated to each user within a given time slot. 
Fig.~\ref{fig:OMA_NOMA_Rate_Downlink}(b) presents the ER performance of a two-user NOMA system. It is evident that at low SNRs, $\mathrm{U}_1$ achieves a higher ER than $\mathrm{U}_2$.
As the SNR increases, the ER of $\mathrm{U}_1$ gradually converges to a constant ceiling, while that of $\mathrm{U}_2$ increases monotonically and approaches a fixed slope.
This phenomenon can be attributed to the fact that at low SNRs, $\mathrm{U}_1$ attains a higher ER as a result of receiving a larger portion of the transmit power.
Subsequently, since the SIC process eliminates the interference from $\mathrm{U}_1$’s signal, the ER of $\mathrm{U}_2$ increases monotonically with the SNR. In contrast, $\mathrm{U}_1$ treats $\mathrm{U}_2$’s signal as interference, causing its ER to reach a ceiling.

\begin{figure}[t]
\centering
\includegraphics[width=\linewidth]{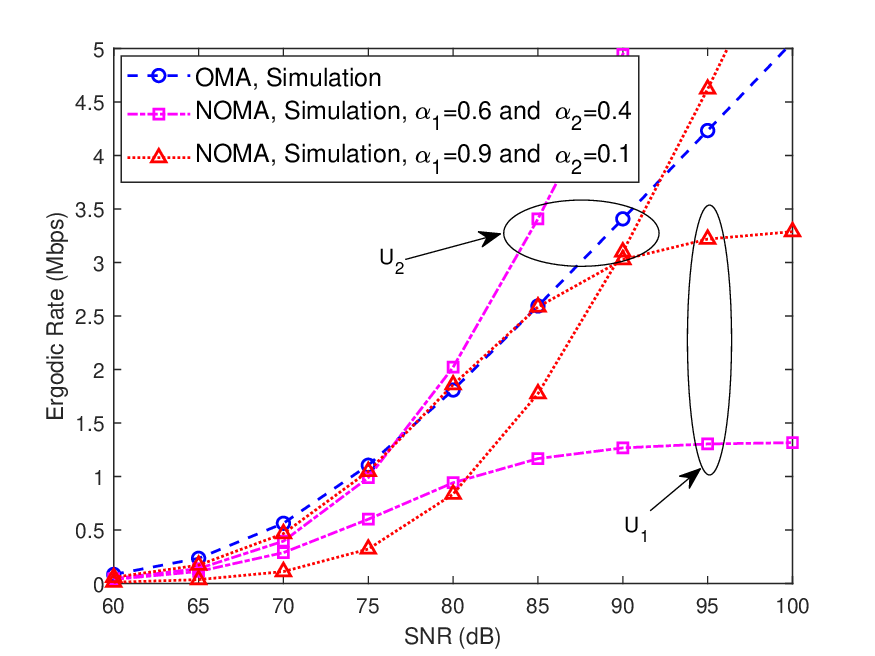}
\caption{Ergodic rate versus SNR for downlink OMA and NOMA systems, where $M=2$.}
\label{fig:OMAandNOMA_rate_simulation_compare}
\end{figure}
Fig.~\ref{fig:OMAandNOMA_rate_simulation_compare} illustrates the ER versus SNR for users in the downlink of OMA and NOMA systems. 
In the NOMA system, two distinct power allocation schemes are considered: $\alpha_1=0.6$, $\alpha_2=0.4$, and $\alpha_1=0.9$, $\alpha_2=0.1$. 
As shown in the figure, under the first power allocation scheme, the ER performance gap between the two users is more pronounced, with $\mathrm{U}_2$ exhibiting a significantly higher ER than $\mathrm{U}_1$. 
In contrast, in the second scheme, $\mathrm{U}_1$ achieves a higher ER in the low-SNR region owing to its higher power allocation. 
Furthermore, in high-SNR regime, the slope of the ER curve for OMA users remains between those of the two NOMA users. 
It can also be observed that $\mathrm{U}_2$ in the NOMA system consistently outperforms both the OMA user and $\mathrm{U}_1$ at high SNR. 
This is primarily because, in the high-SNR regime, the interference from $\mathrm{U}_1$ becomes negligible and can be regarded as weak noise, allowing $\mathrm{U}_2$ to more effectively exploit the full bandwidth and achieve superior performance compared to OMA users.

\subsection{Uplink}
\begin{figure}[t]
\centering
\includegraphics[width=0.9\linewidth]{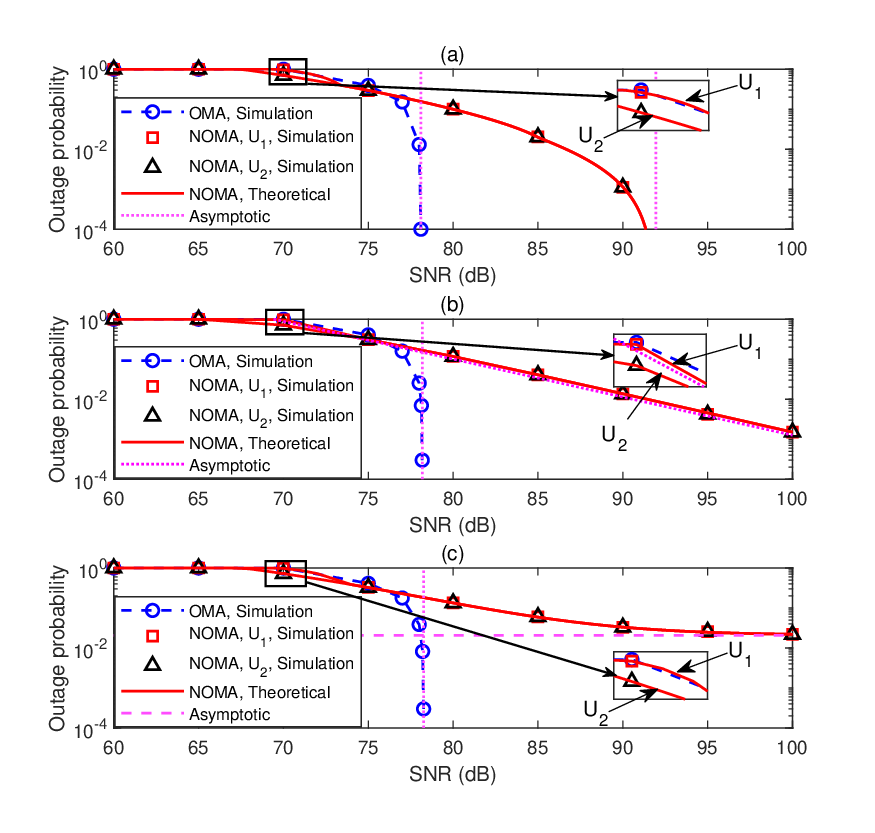}
\caption{OP versus SNR for uplink OMA and NOMA systems.}
\label{fig:Uplink_OP_NOMA_OMA}
\end{figure}
In Fig.~\ref{fig:Uplink_OP_NOMA_OMA} illustrates OP versus SNR for OMA and NOMA systems in the uplink.
In Fig.~\ref{fig:Uplink_OP_NOMA_OMA}(a), where $\tilde{R} = 0.99~\text{Mbps}$, the OP of both users decreases with increasing SNR and eventually approaches zero, which agrees with Proposition~\ref{Proposition:uplink-NOMA-OP-infty} where $\tilde{R} \in (0,1)$.
In Fig.~\ref{fig:Uplink_OP_NOMA_OMA}(b), where $\tilde{R} = 1~\text{Mbps}$, as the SNR increases, both users exhibit a linear decay in the OP when plotted on a logarithmic axis. 
In Fig.~\ref{fig:Uplink_OP_NOMA_OMA}(c), where $\tilde{R} = 1.01~\text{Mbps}$, all OP curves converge to a constant level rather than decreasing with SNR, which is consistent with Proposition~\ref{Proposition:uplink-NOMA-OP-infty} where $\tilde{R} \in (1,\infty)$.
Furthermore, the results demonstrate that the outage performance of the NOMA system is generally inferior to that of the OMA system in this scenario.

\begin{figure}[t]
\centering
\includegraphics[width=\linewidth]{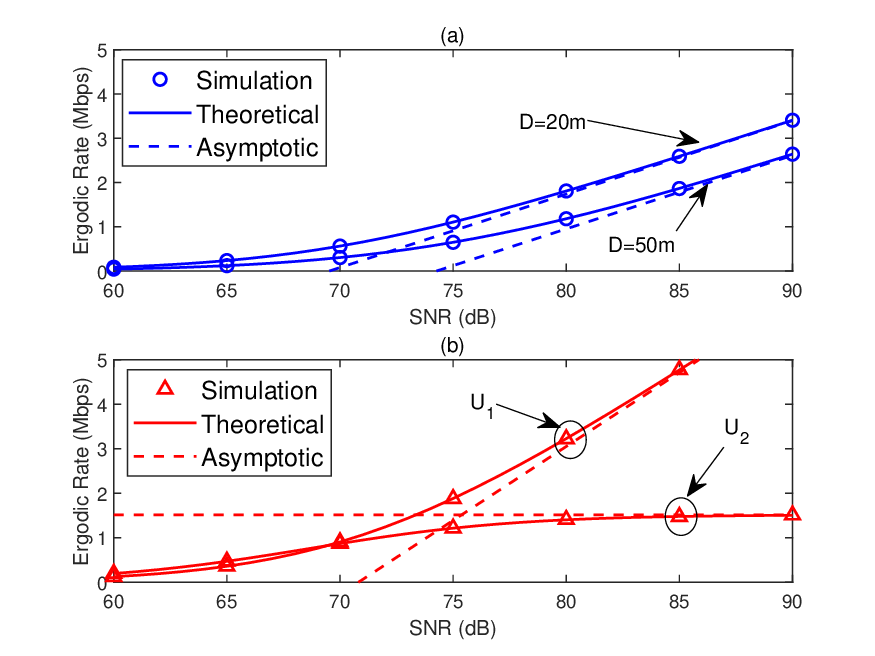}
\caption{Ergodic rate versus SNR for uplink OMA and NOMA systems. (a) illustrates the OMA system performance for $D=20$m and $D=50$m, respectively. (b) illustrates the NOMA system in a two-user scenario.}
\label{Uplink_rate_NOMA_OMA}
\end{figure}
Fig.~\ref{Uplink_rate_NOMA_OMA} illustrates the ER versus SNR for users in OMA and NOMA systems under the uplink scenario. In Fig.~\ref{Uplink_rate_NOMA_OMA}(a), simulations are conducted for the OMA system under different room widths to evaluate the impact of distance parameters on user ER performance. The results indicate that performance with $D=20$m outperforms that with $D=50$m. Moreover, both curves exhibit identical slopes in the high-SNR regime, aligning with the theoretical derivations. This is because larger room sizes increase the average user-antenna distance, resulting in greater path loss. Consequently, the transmission SNR decreases, resulting in a degradation of ER performance. 
Fig.~\ref{Uplink_rate_NOMA_OMA}(b) presents the simulation results for the two-user scenario. It can be observed that the ER of $\mathrm{U}_1$ continuously increases with SNR and eventually converges to a linearly increasing curve with a fixed slope. In contrast, the ER of $\mathrm{U}_2$ gradually stabilizes and converges to a specific value. This difference arises because, in the uplink scenario, the decoding order of SIC is opposite to that in the downlink, leading to an inverse trend of ER variation for the same user under different transmission directions.

\begin{figure}[t]
\centering
\includegraphics[width=\linewidth]{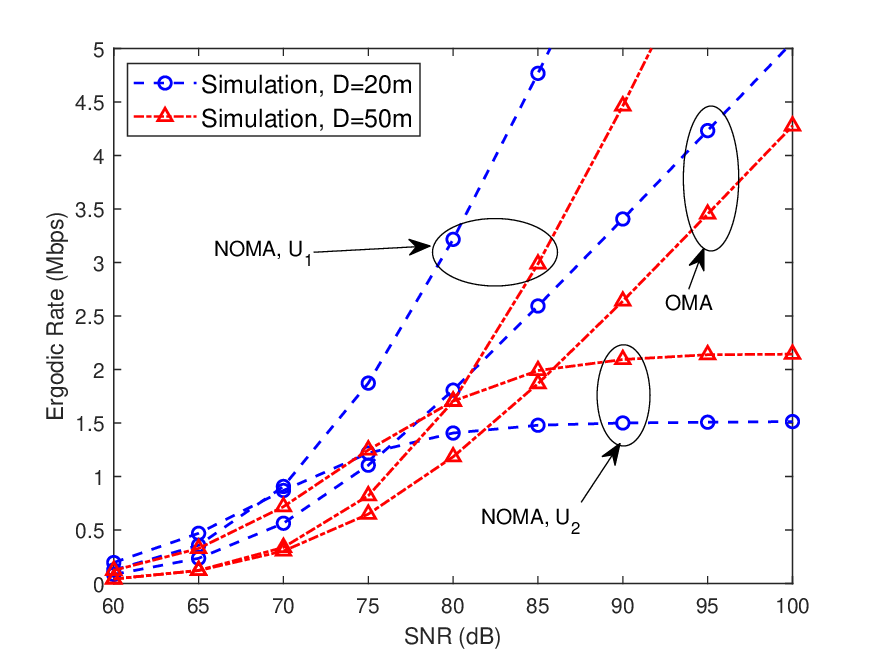}
\caption{Ergodic rate versus SNR for uplink OMA and NOMA systems, where $M=2$.}
\label{Uplink_rate_NOMAandOMA}
\end{figure}
Fig. \ref{Uplink_rate_NOMAandOMA} illustrates the ER of users in OMA and NOMA systems under different distance parameters for the uplink. The simulation results indicate that the ER performance of both OMA and NOMA users at $D=20$ m is significantly superior to that at $D=50$ m. This is attributed to the fact that increasing distance exacerbates channel attenuation, leading to performance degradation. Furthermore, in the low-SNR region, the ER of NOMA users is generally higher than that of OMA users. However, in the high-SNR regime, the overall performance of OMA users outperforms that of certain NOMA users. This phenomenon arises because the application of SIC in uplink NOMA causes the ER of all users, except for the last decoded one, to converge to a constant value.

Finally, the results demonstrates the strong agreement between the analytical results and the Monte Carlo simulations, validating the accuracy of the derived expressions.

\section{Conclusions}
This paper focuses on the PA system, an emerging technology that has recently attracted significant research interest. 
First, we conduct a performance analysis of OMA and NOMA schemes in the downlink scenario. 
The results demonstrate that in the downlink fixed-rate scenario, which scheme achieves better outage performance is determined by system parameters.
Second, a similar investigation is conducted for the uplink, revealing that the OP of NOMA users exhibits different decay rate in the high-SNR regime, depending on the specific rate thresholds. 
Finally, for both downlink and uplink transmissions, which scheme achieves better rate performance in the low-SNR regime still depends on system parameters.
However, in the high-SNR regime, the rate performance of OMA users is generally superior to that of NOMA users. 
Given these performance trade-offs, the selection of multiple access techniques for multi-room PA systems, along with the joint optimization of antenna positions and quantities, represents a significant direction for future research.

\begin{appendices}
\section{Proof of Theorem \ref{Theorem:downlink-NOMA-ER}}\label{Appen-theorem-1}
\renewcommand{\theequation}{\thesection.\arabic{equation}}
\setcounter{equation}{0}
In the NOMA scenario, the ER of the user can be given by
\begin{equation}\label{}
\begin{split}
\mathbb{E} \left (R_{m,m}^{D,NOMA}  \right ) =\mathbb{E}\left ( \log_2\left ( 1+   \mathrm{SINR}_{m,m}^{D,NOMA}   \right )  \right ).
\end{split}
\end{equation}
\subsection{Upper Bound of ER}
Due to the difficulty of obtaining a closed-form solution directly, we employ Jensen's inequality to derive its upper bound. By taking the derivative of $f_1(y)=\log_2\left ( 1+  \mathrm{SINR}_{m,m}^{D,NOMA}   \right )$, it is observed to be a concave function, hence it can be given by
\begin{equation}\label{}
\begin{split}
\mathbb{E}\left ( f_1(y)  \right ) \le  \log_2\left ( 1+   \mathbb{E}\left (f_2(y)\right )   \right )  ,
\end{split}
\end{equation}
where $f_2(y) = \mathrm{SINR}_{m,m}^{D,NOMA}$. Then we can obtain the following expression as $\mathbb{E}\left (f_2(y)\right )= \int_{\frac{\eta}{d^2+\frac{D^2}{4} } }^{\frac{\eta}{d^2} } \frac{y \rho_D \alpha_m}{y\rho_D  {\textstyle \sum_{j=m+1}^{M}\alpha_j}+M }
f_{Y_m}(y)dy$.
After defining $t = \sqrt{\frac{\eta}{y} - d^2}$, the expression can be given by
$
\mathbb{E}\left (f_2(y)\right )= \Lambda_3 \int_{0}^{\frac{D}{2} } \frac{t^{M-m}\left ( 1-\frac{2}{D}t  \right )^{m-1} }
{t^2 + d^2 + \frac{\eta \rho_D  {\textstyle \sum_{j=m+1}^{M}}\alpha_j }{M} } dt.
$
Let $\zeta = \sqrt{d^2 + \frac{\eta \rho_D \sum_{j=m+1}^{M} \alpha_j}{M}}$, then it can be further simplified as
\begin{equation}\label{}
\begin{split}
\mathbb{E}\left (f_2(y)\right )= \Lambda_3\sum_{k=0}^{m-1}\begin{pmatrix}m-1\\k\end{pmatrix}\left ( -\frac{2}{D}  \right )^k
\underbrace{\int_{0}^{\frac{D}{2} } \frac{t^{M-m+k}}{t^2+\zeta^2}dt }_{\xi_k}.
\end{split}
\end{equation}
Let $n_k = M - m + k$, and for $\xi_k$, we consider two separate cases:

\textit{Case 1:} When $n_k=2l$, $n_k$ is even. We have
$
\xi_k=\sum_{r=0}^{l-1} \frac{ (-1)^r \zeta^{2r}\left ( \frac{D}{2}  \right )^{2(l-r)-1} }{2(l-r)-1} + (-1)^l \zeta^{2l-1} \arctan\left ( \frac{D}{2\zeta}  \right ).
$

\textit{Case 2:} When $n_k=2l+1$, $n_k$ is odd. We have
$
\xi_k=\sum_{r=1}^{l-1} \frac{(-1)^r \zeta^{2r} \left ( \frac{D}{2}  \right )^{2(l-r)} }{2(l-r)}+(-1)^l\frac{\zeta^{2l}}{2} \mathrm{ln}\left ( \frac{D^2}{4\zeta^2}  +1\right ).
$

\subsection{Lower Bound of ER}
We define $f_3(y) = \frac{1}{f_2(y)}$, and thus we can be given by
$\mathbb{E}\left ( \log_2\left ( 1+ f_2(y)\right )   \right )=\mathbb{E}\left ( \log_2\left ( 1+ \frac{1}{f_3(y)} \right )   \right )$.
Since $\log_2\left( 1 + \frac{1}{f_3(y)} \right)$ is a convex function with respect to $f_3(y)$, by applying Jensen's inequality, we obtain a lower bound for the ER as 
$\mathbb{E}\left ( \log_2\left ( 1+ \frac{1}{f_3(y)} \right )   \right ) \ge  \log_2\left ( 1+ \frac{1}{\mathbb{E}\left (  f_3(y) \right )} \right )$.
For $\mathbb{E}\left( f_3(y) \right)$, it can be expressed as 
$\mathbb{E}\left (  f_3(y) \right )=\frac{\sum_{j=m+1}^{M}\alpha_j}{\alpha _m} +\frac{M}{\rho_D \alpha _m} \mathbb{E}\left ( \frac{1}{y}  \right ) $.
Therefore, we only need to compute $\mathbb{E}\left( \frac{1}{y} \right)$, which is given by 
$\mathbb{E}\left ( \frac{1}{y}  \right )
=M \begin{pmatrix}M-1\\m-1\end{pmatrix}\int_{\frac{\eta}{d^2+\frac{D^2}{4} } }^{\frac{\eta}{d^2} } \frac{1}{y}
\frac{\eta}{Dy^2\sqrt[]{\frac{\eta}{y}-d^2 } } \left ( 1- \frac{2}{D} \sqrt[]{\frac{\eta}{y}-d^2 } \right )^{m-1} \\
 \left ( \frac{2}{D} \sqrt[]{\frac{\eta}{y}-d^2 } \right )^{M-m} dy$.
Let $t = \sqrt{\frac{\eta}{y} - d^2}$, which can be simplified as follows
\begin{equation}\label{}
\begin{split}
\mathbb{E}\left ( \frac{1}{y}  \right ) =  \Lambda_4
 \sum_{k=0}^{m-1}\begin{pmatrix}m-1\\k\end{pmatrix} \left ( -\frac{2}{D}  \right )^k \\
\times\left ( \frac{\left ( \frac{D}{2}  \right )^{M-m+k+3} }{M-m+k+3}
+\frac{\left ( \frac{D}{2}  \right )^{M-m+k+1} d^2 }{M-m+k+1}   \right )  .
\end{split} 
\end{equation}
where $\Lambda_4=\frac{M}{\eta} \begin{pmatrix}M-1\\m-1\end{pmatrix}\left ( \frac{2}{D}  \right )^{M-m+1}$, $a$.
This completes the proof.

\section{Proof of Theorem \ref{OP:uplink-NOMA}}\label{Appendix:OP-uplink-NOMA}
\renewcommand{\theequation}{\thesection.\arabic{equation}}
\setcounter{equation}{0}

For $\mathrm{U_1}$, based on \eqref{eq:outage_U1_uplink_NOMA}, \eqref{Eq:uplink-NOMA-OP} and \eqref{Uplink:NOMA-OP-omga}, its OP can be given by  
\begin{equation}\label{}
\begin{split}
\mathbb{P}_1^{U,NOMA}=1-\underbrace{\int_{\omega_1}^{b} \int_{\omega_2}^{b}f_{|h_1|^2|h_2|^2}(y_1,y_2) dy_1 dy_2}_{I},
\end{split}
\end{equation}
where $a=\frac{\eta}{d^2+\frac{D^2}{4}}$, $b=\frac{\eta}{d^2}$, $\omega_1=\frac{2^{\tilde{R}_1} -1}{\rho_U}$ and $\omega_2=\left ( y_1+\frac{1}{ \rho_U}  \right ) \left ( 2^{\tilde{R}_2} -1\right )$. The joint distribution of the order statistics with absolutely continuous distribution is given by $f_{|h_1|^2,|h_2|^2}(y_1,y_2)=2f_Y(y_1)f_Y(y_2)$.
In this part, the integration with respect to $y_2$ is first carried out. Since $\text{U}_2$ is regarded as the strong channel user, the integration range of $y_2$ in the inner integral is determined as $y_2 \in \left[\max(a, \omega_2, y_1), b \right]$. Accordingly, we can derive as
\begin{equation}
I = \frac{4\eta}{D^2} \int_{\omega_1}^{b}\frac{\sqrt[]{\frac{\eta}{\max(a,\omega_2,y_1)}-d^2 } }{y_1^2 \sqrt[]{\frac{\eta}{y_1}-d^2 } } dy_1.
\end{equation}
It can be observed that when $\omega_1 \ge b$, $I=0$ holds, and therefore $\mathbb{P}_1^{U,NOMA}=1$. When $\omega_1 < b$, the value of $\max(a, \omega_2, y_1)$ needs to be examined in more detail. For instance, if $y_1 < \omega_2$, the inequality $\left ( 2-2^{\tilde{R}_2}\right)y_1 < \frac{2^{\tilde{R}_2}-1}{\rho_U}$ is satisfied, where $c$ is defined as $c=\frac{2^{\tilde R_2}-1}{\rho_U (2-2^{\tilde R_2})}$. Consequently, the analysis can be classified into four cases as follows.

\textit{Case 1:} When $2-2^{\tilde{R}_2} > 0$, i.e., $\tilde{R}_2 \in (0,1)$,  the value of $\max(a, \omega_2, y_1)$ can written as
\begin{equation}
\max(a, \omega_2, y_1) = 
\begin{cases} 
y_1, &  \max(a,c,\omega_1) \le y_1 \le b \\ 
& \text{and} \ c<b , \\ 
\omega_2, &  \max(a,\omega_1)\le y_1 \le \min(b,c)\\ 
& \text{and} \ c>\max(a,\omega_1).
\end{cases}
\end{equation}
In the first part, $\max(a, \omega_2, y_1)= y_1$, we can derive as
$
I_1 =  \frac{4\eta}{D^2} \left ( \frac{1}{\max(\omega_1, c, a)} - \frac{1}{b}  \right ) .
$
In the second part, $\max(a, \omega_2, y_1)= \omega_2$, we can derive as
\begin{equation}
\begin{split}
I_2 = &\frac{4\eta}{D^2} \int_{\max(\omega_1, a)}^{\min(b, c)}\frac{\sqrt[]{\frac{\eta}{\min(\omega_2,b)}-d^2 } }{y_1^2 \sqrt[]{\frac{\eta}{y_1}-d^2 } } dy_1   \\
\approx & \frac{2\eta \left (  \min(b,c)-\max(\omega_1, a)\right ) }{D^2} \sum_{i=1}^{n} w_i \ \sqrt[]{1-t_i^2} f\left ( x_i \right ) ,
\end{split}
\end{equation}
where $w_i=\frac{\pi}{n}$, $t_i=\mathrm{cos}\left( \frac{2i-1}{2n} \pi\right)$, $x_i=\frac{\min(b,c)-\max(\omega_1, a)}{2}t_i+\frac{\min(b,c)+\max(\omega_1, a)}{2}$, $f\left( x _i \right)=\frac{\sqrt[]{\frac{\eta}{\min(\omega_2(x_i),b)}-d^2 } }{x_i^2\sqrt[]{\frac{\eta}{x_i}-d^2 } }$, $\omega_2(x_i)=\left ( x_i+\frac{1}{\rho_U}\right)\left ( 2^{\tilde{R}_2}-1\right )$. Finally, by combining both parts, the OP of $\mathrm{U}_1$ is derived as  
\begin{equation}
\begin{split}
\mathbb{P}_1^{U,NOMA} \approx 1 - \left (I_1 + I_2 \right).
\end{split}
\end{equation}

\textit{Case 2:} When $2-2^{\tilde{R}_2} = 0$, i.e., $\tilde{R}_2= 0$, $\max(a, \omega_2, y_1)= \omega_2$ always holds. Therefore, we can derive as
$I=I_3 \approx \frac{2\eta \left (  b-\max(\omega_1, a)\right ) }{D^2} \sum_{i=1}^{n} w_i \ \sqrt[]{1-t_i^2} f\left ( x_i \right )
$, similarly, 
where $x_i=\frac{b-\max(\omega_1, a)}{2}t_i+\frac{b+\max(\omega_1, a)}{2}$. Finally, by combining both parts, the OP of $\mathrm{U}_1$ is derived as  
\begin{equation}
\begin{split}
\mathbb{P}_1^{U,NOMA} \approx 1 - I_3.
\end{split}
\end{equation}

\textit{Case 3:} When $2-2^{\tilde{R}_2} < 0$, i.e., $\tilde{R}_2 \in (1,\infty)$, the value of $\max(a, \omega_2, y_1)$ can written as
\begin{equation}
\max(a, \omega_2, y_1)= 
\begin{cases} 
y_1, &   \max(a,\omega_1)\le y_1 \le \min(b,c)\\ 
& \text{and} \ c>\max(a,\omega_1), \\
\omega_2, &\max(a,c,\omega_1) \le y_1 \le b \\ 
& \text{and} \ c<b.
\end{cases}
\end{equation}
In the first part, $\max(a, \omega_2, y_1)= y_1$, we can similarly derive as
$
I_4 = \frac{4\eta}{D^2} \left ( \frac{1}{\max(\omega_1, a)} - \frac{1}{\min(b, c)}  \right )$.
In the second part, $\max(a, \omega_2,y_1)= \omega_2$, we can derive as
$
I_5 \approx  \frac{2\eta \left (  b-\max(\omega_1, c, a)\right ) }{D^2} \sum_{i=1}^{n} w_i \ \sqrt[]{1-t_i^2} f\left ( x_i \right )
$,
where $x_i=\frac{b-\max(\omega_1, c, a)}{2}t_i+\frac{b+\max(\omega_1, c, a)}{2}$. Finally, by combining both parts, the OP of $\mathrm{U}_1$ is derived as  
\begin{equation}
\begin{split}
\mathbb{P}_1^{U,NOMA} \approx 1 - \left (I_4 + I_5 \right).
\end{split}
\end{equation}

For $\mathrm{U_2}$, based on \ref{eq:outage_U2_uplink_NOMA}, we can obtain
\begin{equation}\label{}
\begin{split}
\mathbb{P}_2^{U,NOMA}= &2\int_{a}^{b} \int_{y_1}^{\min(\omega_2,b)} f_Y(y_1)f_Y(y_2)dy_2dy_1 \\
= & \frac{4\eta}{D^2} \int_{a}^{b}\frac{1}{y_1^2} - \frac{\sqrt[]{\frac{\eta}{\min(\omega_2,b)}-d^2 } }{y_1^2 \sqrt[]{\frac{\eta}{y_1}-d^2 } } dy_1 ,
\end{split}
\end{equation}
where $a \le y_1 \le y_2 \le \min(\omega_2,b)$. For $y_1 < \omega_2$, it can be given by $\left ( 2-2^{\tilde{R}_2}\right)y_1 < \frac{2^{\tilde{R}_2}-1}{\rho_U}$. Then, we can obtain
\begin{equation}\label{}
y_1 \in
\begin{cases}
\left [ a, \min(b, c) \right ] , & 0 < \tilde{R}_2 < 1,\ c>a,  \\
\left [ a, b \right ] , & \tilde{R}_2=0,  \\
\left [ \max(a, c), b \right ] , & 1 < \tilde{R}_2 ,\ c<b,  \\
\end{cases}
\end{equation}
where $c=\frac{2^{\tilde R_2}-1}{\rho_U (2-2^{\tilde R_2})}$. Similarly, it can also be divided into three cases for derivation.

\textit{Case 1}: $y_1 \in \left [ a, \min(b, c) \right ]$, we can given by
$
\mathbb{P}_2^{U,NOMA} \approx \frac{4\eta }{D^2} \left (  \frac{1}{a}-\frac{1}{\min(b,c)} \right )   
-  \frac{2\eta \left (  \min(b,c)-a\right ) }{D^2} \sum_{i=1}^{n} w_i \ \sqrt[]{1-t_i^2} f\left ( x_i \right ) ,
$
where $w_i=\frac{\pi}{n}$, $t_i=\mathrm{cos}\left( \frac{2i-1}{2n} \pi\right)$, $x_i=\frac{\min(b,c)-a}{2}t_i+\frac{\min(b,c)+a}{2}$, $f\left( x _i \right)=\frac{\sqrt[]{\frac{\eta}{\min(\omega_2(x_i),b)}-d^2 } }{x_i^2\sqrt[]{\frac{\eta}{x_i}-d^2 } }$.

\textit{Case 2}: $y_1 \in \left [ a, b \right ]$, we can given by
$
\mathbb{P}_2^{U,NOMA} \approx 1 -  \frac{2\eta \left (  b-a\right ) }{D^2} \sum_{i=1}^{n} w_i \ \sqrt[]{1-t_i^2} f\left ( x_i \right ) ,
$
where $x_i=\frac{b-a}{2}t_i+\frac{b+a}{2}$.

\textit{Case 3}: $y_1 \in \left [ \max(a, c), b \right ]$, we can given by
$
\mathbb{P}_2^{U,NOMA} \approx \frac{4\eta }{D^2} \left (  \frac{1}{\max(a,c)}-\frac{1}{b} \right ) -  \frac{2\eta \left (  b-\max(a,c)\right ) }{D^2} \sum_{i=1}^{n} w_i \ \sqrt[]{1-t_i^2} f\left ( x_i \right ) ,
$
where $x_i=\frac{b-\max(a,c)}{2}t_i+\frac{b+\max(a,c)}{2}$.
This completes the proof.

\section{Proof of Proposition \ref{Proposition:uplink-NOMA-OP-infty}}\label{Appendix:OP-uplink-NOMA-infty}
\renewcommand{\theequation}{\thesection.\arabic{equation}}
\setcounter{equation}{0}

For $\tilde{R}_2 \in (0,1)$, in the high-SNR regime, we can derive that $\lim_{\rho_U \to \infty}\max(\omega_1, c, a)=a$. Consequently, we have $\lim_{\rho_U \to \infty}I_1=1$. This ultimately leads to $\lim_{\rho_U \to \infty}\mathbb{P}_1^{U,NOMA}=0$. Similarly, it follows that $\lim_{\rho_U \to \infty}\mathbb{P}_2^{U,NOMA}=0$.

For $\tilde{R}_2 = 0$, by setting $\varepsilon = \frac{1}{\rho_U}$, $I(\varepsilon)=\int_{a}^{b-\varepsilon }f(y_1,\varepsilon)dy_1$, where $f(y_1,\varepsilon)=\frac{\sqrt[]{\frac{\eta}{y_1+\varepsilon }-d^2 } }{y_1^2 \sqrt[]{\frac{\eta}{y_1}-d^2 } }$. Since $I(0)=1$, we obtain
\begin{equation}\label{}
\begin{split}
\mathbb{P}_1^{U,NOMA,\infty}= \lim_{\varepsilon \to 0} -\frac{4\eta}{D^2}\left ( I(\varepsilon ) -I(0)\right ).
\end{split}
\end{equation}
For $I(\varepsilon ) -I(0)$, we obtain
$I(\varepsilon ) -I(0) = -\int_{b-\varepsilon}^{b}f(y_1,\varepsilon)dy_1 +\int_{a}^{b}\left ( f(y_1,\varepsilon)-f(y_1,0) \right ) dy_1$.
Since $\frac{\eta}{b}-d^2=0$, $f(y_1,\varepsilon)$ exhibits a singularity at $y_1=b$. By setting $g(y_1)=\frac{\eta}{b}-d^2$ and $u=b-y$, and performing a Taylor expansion around $y_1=b$, we obtain $g(y_1) = \frac{\eta u}{b^2} + O(u^2)$.
Consequently, $\sqrt{\frac{\eta}{y_1}-d^2}=\frac{\sqrt[]{u} }{b}\sqrt[]{t}$. Similarly, we have
$\sqrt{\frac{\eta}{y_1+\varepsilon}-d^2}\approx d \ \sqrt[]{\frac{u-\varepsilon}{b} }$, $y_1^2\approx b^2$. Hence, $f(b-u,\varepsilon)=\frac{1}{b^2}\sqrt[]{1-\frac{\varepsilon }{u} } \ (u>\varepsilon)$. Then, by utilizing the Mean Value Theorem for Integrals, we can obtain
$-\int_{b-\varepsilon}^{b}f(y_1,\varepsilon)dy_1 \approx -\frac{\varepsilon }{b^2}$.
In $u \in [\varepsilon,\delta ]$, where $\delta$ is a small constant independent of $\varepsilon$, this term provides the dominant contribution. Thus, we can obtain $\int_{a}^{b}\left ( f(y_1,\varepsilon)-f(y_1,0) \right ) dy_1 \approx -\frac{\varepsilon}{2b^2}\ln \left (\frac{1}{\varepsilon} \right)$. 
Finally, since the similarity of OP formulas for $\mathrm{U}_1$ and $\mathrm{U}_2$ in the high-SNR regime, we have
\begin{equation}\label{}
\begin{split}
\mathbb{P}_1^{U,NOMA,\infty}=\mathbb{P}_2^{U,NOMA,\infty} \approx \frac{2d^4}{D^2\eta}\frac{\ln \rho_U}{\rho_U}.
\end{split}
\end{equation}
Then, $\mathbb{P}_{\{ 1,2\}}^{U,NOMA}$ decays as $\rho_U^{-1}$ in the high-SNR regime since the term $\frac{\ln \rho_U}{\rho_U}$ is dominated by the linear growth of $\rho_U$ as $\rho_U \to \infty$, effectively yielding a constant slope of $-1$ in the log-log scale.

For $\tilde{R}_2 \in (1,\infty)$, we have $\lim_{\rho_U \to \infty}\max(\omega_1, c, a)=a$. Hence, $\lim_{\rho_U \to \infty}I_4=0$, $I_9=\lim_{\rho_U \to \infty}I_5\approx \frac{2\eta (b-a ) }{D^2} \sum_{i=1}^{n} w_i \ \sqrt[]{1-t_i^2} f\left ( x_i \right ) $, where $f(x_i)=\frac{\sqrt[]{\frac{\eta}{\min \left ( x_i\left ( 2^{\tilde{R}_2}-1,b\right)\right) }-d^2 } }{x_i^2 \sqrt[]{\frac{\eta}{x_i}-d^2 } }$, $x_i=\frac{b-a}{2}t_i+\frac{b+a}{2}$. Finally, we can derive that $\lim_{\rho_U \to \infty}\mathbb{P}_1^{U,NOMA}=\lim_{\rho_U \to \infty}\mathbb{P}_2^{U,NOMA}=1-I_9$.
This completes the proof.

\section{Proof of Theorem \ref{theorem-uplink-NOMA-ER-two-User}}\label{Appen-theorem-uplink-NOMA-ER-two-User}
\renewcommand{\theequation}{\thesection.\arabic{equation}}
\setcounter{equation}{0}
For $\mathrm{U}_1$, based on the ordered channel distribution, the PDF of $\mathrm{U}_1$ is given by $f_1(y)=2f_{Y}(y)\left ( 1-F_{Y}(y) \right )$. Therefore, we have
\begin{equation}\label{}
\begin{split}
\mathbb{E}\left ( R_1^{U,NOMA} \right ) = \frac{4\eta}{D^2}\underbrace{  \int_{a}^{b}\log_2\left ( 1+\rho_U y \right ) \frac{1}{y^2}dy }_{\Xi_1},
\end{split}
\end{equation}
where $a= \frac{\eta}{d^2 + \frac{D^2}{4}}$, $b= \frac{\eta}{d^2}$. By integrating $\Xi_1$, the result can be obtained as $\Xi_1=G(b)-G(a)$, where $G(x)=\frac{\ln(1+\rho_U x)}{x}-\rho_U \ln x+\rho_U\ln(1+\rho_U x)$.

For $\mathrm{U}_2$, the joint PDF is given by $f_{|h_1|^2 |h_2|^2}(y_1, y_2)=2f_{Y}(y_1)f_{Y}(y_2)$, which can be expressed as
\begin{equation}
\begin{split}
&\mathbb{E}\left ( R_2^{U,NOMA} \right ) = \int_{a}^{b}\int_{a}^{b}\log_2\left ( 1+\frac{y_2}{y_1+\frac{1}{\rho_U} }  \right )\\
&\times 2f_1(y_1)f_2(y_2)dy_1dy_2 \\
&=\frac{2}{\ln 2}\underbrace{ \int_{a}^{b}\underbrace{\int_{y_1}^{b} \ln \left ( 1+\frac{y_2}{y_1+\frac{1}{\rho_U} } \right )d F_Y(y_2)}_{\Xi_2}f_Y(y_1) dy_1}_{\Xi_3}.
\end{split}
\end{equation}
where $a= \frac{\eta}{d^2 + \frac{D^2}{4}}$, $b= \frac{\eta}{d^2}$. For $\Xi_2$, we obtain 
$\Xi_2=\frac{2}{D} \ln\left ( 1+\frac{y_1 }{y_1 + \frac{1}{\rho_U } }  \right ) \sqrt[]{\frac{\eta}{y_1}-d^2 }
+ \frac{4}{D}\sqrt[]{\frac{\eta}{y_1 + \frac{1}{\rho_U } } + d^2 } \\
\arctan \left ( \frac{\sqrt[]{\frac{\eta}{y_1}-d^2 } }{\sqrt[]{\frac{\eta}{y_1 + \frac{1}{\rho_U } } + d^2 }}  \right ) 
-   \frac{4d}{D} \arctan\left ( \frac{\sqrt[]{\frac{\eta}{y_1}-d^2 }}{d}  \right )$.
Then, for $\Xi_3$, we have $\Xi_3=\int_{a}^{b}f_Y(y_1)\Xi_2dy_1$.
Due to the complexity of the integrand, we adopt the Chebyshev--Gauss quadrature method by setting $y_1 = \frac{b-a}{2}t + \frac{b+a}{2}$. Finally, we obtain the expression
\begin{equation}\label{}
\begin{split}
&\mathbb{E}\left ( R_2^{U,NOMA} \right )  \approx\frac{b-a}{\ln2} \sum_{i=1}^{n}w_i \sqrt{1-t_i^2} H\left( x _i \right) ,
\end{split}
\end{equation}
where $w_i=\frac{\pi}{n}$, $t_i=\mathrm{cos}\left( \frac{2i-1}{2n} \pi\right)$, $x_i=\frac{b-a}{2}t_i+\frac{b+a}{2}$, $H\left( x _i \right)=f_Y(x_i)\Xi_2(x_i)$.
This completes the proof.

\end{appendices}
\bibliographystyle{IEEEtran}
\bibliography{ref}

\end{document}